\documentclass[11pt,authoryear]{article}
\usepackage[latin9]{inputenc}
\usepackage{geometry}
\usepackage{color}
\usepackage{amsmath}
\usepackage{amsthm}
\usepackage{amssymb}
\usepackage{graphicx}
\usepackage{setspace}
\usepackage{esint}
\makeatletter
  \theoremstyle{definition}
  \newtheorem{defn}{\protect\definitionname}
  \theoremstyle{plain}
  \newtheorem{thm}{\protect\theoremname}
 \ifx\proof\undefined\
   \newenvironment{proof}[1][\proofname]{\par
     \normalfont\topsep6\p@\@plus6\p@\relax
     \trivlist
     \itemindent\parindent
     \item[\hskip\labelsep
           \scshape
       #1]\ignorespaces
   }{%
     \endtrivlist\@endpefalse
   }
   \providecommand{\proofname}{Proof}
 \fi
  \theoremstyle{plain}
  \newtheorem{prop}{\protect\propositionname}

\usepackage{amsfonts}
\usepackage{graphicx}
\usepackage{comment}
\usepackage[titletoc,title]{appendix}
\providecommand{\U}[1]{\protect\rule{.1in}{.1in}}
\RequirePackage[colorlinks,citecolor=blue,urlcolor=blue]{hyperref}
\RequirePackage{hypernat}
\providecommand{\U}[1]{\protect\rule{.1in}{.1in}}
\numberwithin{equation}{section}

\makeatother

\providecommand{\definitionname}{Definition}
\providecommand{\propositionname}{Proposition}
\providecommand{\theoremname}{Theorem}

\begin{document}

\title{\vspace{-1in}
\parbox[c]{1\linewidth}{%
} \bigskip{}
 \\
Dynamic portfolio selection without risk-free assets }

\author{Chi Kin Lam\thanks{Department of Statistics and Actuarial Science, The University of
Hong Kong, Pokfulam Road, Hong Kong, e-mail: cklamsta@hku.hk}, Yuhong Xu\thanks{Mathematical Center for Interdiscipline Research and School of Mathematical
Sciences, Soochow University, Suzhou 215006, People's Republic of
China, e-mail: yhxu@suda.edu.cn.} and Guosheng Yin\thanks{Department of Statistics and Actuarial Science, The University of
Hong Kong, Pokfulam Road, Hong Kong, e-mail: gyin@hku.hk}}

\date{February 13, 2016}

\maketitle
\indent\textbf{Abstract.} We consider the mean--variance portfolio
optimization problem under the game theoretic
framework and without risk-free assets. The problem is solved semi-explicitly by applying the extended
Hamilton--Jacobi--Bellman equation. Although the coefficient of risk aversion in
our model is a constant, the optimal amounts of money invested in
each stock still depend on the current wealth in general. The optimal
solution is obtained by solving a system of ordinary differential
equations whose existence and uniqueness are proved and a numerical
algorithm as well as its convergence speed are provided. Different from portfolio selection with risk-free assets, our value function is quadratic in the current wealth, and the equilibrium allocation is linearly sensitive to the initial wealth. Numerical results show that this model performs better than both the classical one and the variance model in a bull market.\\

\indent\textbf{Key words:} mean--variance portfolio selection, asset
allocation, time-inconsistency, equilibrium control, Hamilton--Jacobi--Bellman
equation. \clearpage{}

\section{Introduction}

The portfolio selection problem deals with how to allocate the wealth
among a set of assets. In his seminal work, Markowitz (1952, 1959)
first proposed the mean--variance (MV) portfolio selection theory
in a single period framework, which is regarded as the cornerstone
in modern finance. In this pioneering theory, an investor aims at
maximizing the mean return and minimizing the variance, which is regarded
as the measurement of risk, at the same time. This problem has multiple
solutions, which comprise the so-called efficient frontier, since
it is a multi-objective optimization problem. In fact, each point
on the efficient frontier is the optimal solution for the single-objective
optimization problem which is to minimize the corresponding variance
subject to a given level of the expected wealth. When short-selling
is allowed and the covariance matrix which is composed by the volatilities
of the stocks is nonnegative definite, the analytic expression of
the mean--variance frontier is derived in Markowitz (1956) and Merton
(1972). In the case where the covariance matrix is non-negative definite,
Perold (1984) describes an algorithm for solving the MV portfolio
selection problem. However, there is a criticism on how the risk is
measured in the original MV framework. For the discussion on the replacement
of the risk measurement, see Markowitz (1959). Besides the variance
of the expected portfolio return, alternative measurements of the
risk such as the semi-variance, the lower partial moment and the downside
risk are proposed for constructing the optimal portfolio, see Konno
and Yamazaki (1991), Markowitz et al. (1993), Zenios and Kang (1993)
and Ogryczak and Ruszczy\'{n}ski (1999).

A main challenge for extending the original single period model to
the multi-period case is the time inconsistency since the Bellman
Optimality Principle is violated. In this case, the optimality of
a control depends on both the current and the initial states. The
concept of the term ``optimality'', as well as ``an optimal control
law'', is therefore unclear. Technically, we cannot apply the dynamic
programming directly to attack this problem.

There are three popular ways for handling a family of the time-inconsistent
problems. The first one, known as the ``pre-commitment'' strategy
in the economics literature, seeks a strategy that optimizes the objective
function at the initial time. Whether it is optimal for the objective
function in the future is disregarded. Here, the interpretation of
``optimal'' is ``optimal from the point of view of the initial time''.
Richardson (1989) and Bajeux-Besnainou and Portait (1998) first develop
a continuous-time version of the MV model under the pre-committed
setting. Another extension to the multi-period version can be found
in Li and Ng (2000). They embed the original time-inconsistent problem
into a class of auxiliary stochastic linear--quadratic (LQ) control
problems. Using the similar technique, Zhou and Li (2000), Lim and
Zhou (2002), Lim (2004), Bielecki et al. (2005), Xia (2005) provide
a solution to the continuous-time MV portfolio selection problem.
With the regime switching, the MV portfolio selection and asset--liability
management problems are studied by Zhou and Yin (2003), Chen et al.
(2008) and Chen and Yang (2011). Dai et al. (2010) provide a pre-committed
strategy when the transaction cost is taken into account.

The second approach for tackling the time inconsistency is that instead
of using strategy that is fixed at the initial time, an investor keeps
updating his wealth allocation in order to optimize the corresponding
objective function at the current time.

The third approach is to treat the time inconsistency seriously. For
this situation, a major challenge is that the dynamic programming
approach cannot be applied directly since the iterated-expectations
property is violated due to the variance term, which is not a linear
function of the expected value of the wealth, involved in the objective
function. Some of the early relevant literatures are Strotz (1955)
and Pollak (1968). In Strotz (1955), the author demonstrates that
if a discount function is applied to consumption plans, a certain
plan which is optimal to an investor at the beginning may not be the
case in the future. However, in certain cases, the strategies developed
in these papers for handling the time inconsistency issue do not exist.
See Peleg and Yaari (1973). In Peleg and Yaari (1973), the time-inconsistent
problems are treated as noncooperative games and the optimal strategies
is described using Nash equilibrium. Within this framework, there
are one player at each time point and every player should find his
own strategy in order to maximize his objective function. In fact,
these players can be viewed as your future incarnations. From this
point of view, Ekeland and Lazrak (2006) and Ekeland and Pirvu (2008)
consider the Merton portfolio management problem in the context of
non-constant hyperbolic discounting in deterministic and stochastic
models respectively. A precise definition of the game theoretic equilibrium
concept in continuous time is provided in these two papers.

Basak and Chabakauri (2010) consider the dynamic mean--variance portfolio
problem in an incomplete-market setting. They derive a recursive formulation
for the mean--variance criteria and obtain the closed-form expression
for its time consistent strategy via the dynamic programming approach.
However, their approach can only be applied to the stochastic control
problem with the MV objective function. For a more general class of
the time inconsistent objective functions, Björk and Murgoci (2010)
and Björk and Murgoci (2014) develop both the discrete-time and continuous-time
theories within a game theoretic framework. They derive an extended
Hamilton--Jacobi--Bellman (HJB) equation and provide the corresponding
verification theorem. As an illustration, besides the MV utility model,
the time-inconsistent control problems with non-exponential discounting
and with the utility function at the terminal time depends on the
current state are also solved using the developed theory. However,
Björk et al. (2014) argued that the optimal control developed in Basak
and Chabakauri (2010), which can be reproduced using the theory in
Björk and Murgoci (2010), is not economically reasonable since it
does not depend on the current wealth state. To construct a more realistic
model, the authors consider the case in which the risk aversion depends
on the current wealth. In particular, if the risk aversion is inversely
proportional to the current wealth state, the optimal amount of money
invested in the risky asset is proportional to the wealth. Under short-selling
prohibition, Bensoussan et al. (2014) study the same problem with
the risk aversion being inversely proportional to the current wealth
in both discrete and continuous time setting and prove that the optimal
control in the discrete time model converges to the one in the continuous
setting.

On the other hand, numerical schemes for determining the pre-commitment
strategy and the time-consistent strategy of a continuous MV asset
allocation problem is proposed by Wang and Forsyth (2010) and Wang
and Forsyth (2011). In their algorithms, any type of constraint can
be applied to the investment behavior. Wang and Forsyth (2012) then
extend the numerical techniques for determining these two policies
in the mean quadratic variation problem.

Besides portfolio selection problems, there are other applications
of the extend HJB equations developed in Björk and Murgoci (2010)
under the mean--variance framework. An equilibrium control for the
asset-liability management problem is derived by Wei et al. (2013).
In addition, the optimal time-consistent investment and reinsurance
strategies using the game theoretic approach are constructed, see
Li and Li (2013), Zeng et al. (2013), Li et al. (2015) and Lin and
Qian (2015).

In this paper, we construct the equilibrium control for the MV asset
allocation problem with multiple assets. We consider three models
and in all of them, the assets an investor can trade are multiple
stocks. In model 1, an additional risk-free bond with a constant interest
rate is included. The objective functions considered in model 1 and
model 3 are the same with the one used in the MV portfolio problem
with constant risk aversion while the one in model 2 only includes
the variance term. In fact, a risk-free asset can be considered as
an asset with zero volatility. From this point of view, model 3 can
be regarded as a generalized version of the ones in Björk and Murgoci
(2010) and Basak and Chabakauri (2010). As the risk aversion goes
to infinity, the equilibrium control derived in model 3 converges
to the one in model 2. Furthermore, although the risk aversion considered
in this paper is a constant, if the risk-free asset is not available,
the optimal amount of money invested in each risky asset still depends
on the current wealth, which is unexpected.

The remainder of the paper is organized as follows. In Section \ref{sec:Problem_Formulation},
we present the formulation of our problem as well as the game theoretic
framework. In Section \ref{sec:3_Models_Description}, we state the
three models with different types of assets and objective functions.
For each model, we state the corresponding extended HJB system. With
suitable \textit{Ansatzs} we can solve each system explicitly. For
model 3, the existence and the uniqueness of the solution is proved.
We also provide a numerical algorithm for computing the solution as
well as its convergence speed. Two special cases are also presented
in Section \ref{sec:3_Models_Description}. In Section \ref{sec:Numerical_Results},
some nature parameter combinations in model 3 is provided. Graphical
illustrations of the three models are also presented for comparison.
Section \ref{sec:Conclusion} concludes this paper. The main technical
proofs of the proposition and theorems are given in appendices.

\section{Problem formulation in a game theoretic framework\label{sec:Problem_Formulation}}

Assume that the state $X_{t}$ (typically the wealth process) at time
$t$ is given by a linear stochastic differential equation:
\[
dX_{t}=\mu(t,X_{t},u_{t})dt+\sigma(t,X_{t},u_{t})dW_{t}
\]
where $\mu,\sigma:[0,T]\times R\times R\to R$ satisfying suitable
conditions such that the stochastic differential equation has a unique
solution.

We first recall the problem formulation from Björk and Murgoci (2010).
For deterministic functions $F(x,y)$ and $G(x,y)$, we consider a
reward function of the form
\[
J(t,x,u)=\mathrm{E}_{t,x}\{F(x,X_{T}^{u})\}+G\{x,\mathrm{E}_{t,x}(X_{T}^{u})\}
\]
where $(t,x)$ is the fixed initial point of time and wealth. \textcolor{black}{It
is pointed out in Björk and Murgoci (2010) that the optimization problem
for maximizing this reward function does not satisfy the Bellman optimality
principle due to the dependence on initial state and the appearance
of the second term which is a nonlinear function of the expectation
and thus is a time-inconsistent problem. Dynamic programming is therefore
not available for solving this problem. }

\textcolor{black}{We can formulate the problem in the game theoretic
framework established in Björk and Murgoci (2010) and construct a
time-consistent optimal strategy rather than a precommited one. Within
this framework, the optimization problem is treated as a non-cooperate
game and at each point of time $t$, there is a player $t$ which
can be regarded as an incarnation of the investor. Then the optimal
time consistent strategy $\hat{\mathbf{u}}$ is defined as: for an
arbitrary time point $t$, the optimal strategy for player $t$ is
$\hat{\mathbf{u}}(t,\cdot)$ suppose that each player $s$ where $s>t$
uses the strategy $\hat{\mathbf{u}}(s,\cdot)$. }

\textcolor{black}{We now provide a formal definition of equilibrium
control adopted in this paper. This definition is given by Björk and
Murgoci (2010). }
\begin{defn}
(Equilibrium Control Law). An admissible control law $\hat{\mathbf{u}}$
is called equilibrium control if for every admissible control law
$\mathbf{u}$ valued in $R$ and $h>0$,
\[
\mathbf{u}_{h}(s,y)=\begin{cases}
\mathbf{u}, & \mbox{for \ensuremath{t\le s<t+h,\quad y\in R^{n}}}\\
\hat{\mathbf{u}} & \mbox{for \ensuremath{t+h\le s\le T,\quad y\in R^{n}}},
\end{cases}
\]
such that
\[
\liminf_{h\to0^{+}}\frac{J(t,x,\hat{\mathbf{u}})-J(t,x,\mathbf{u}_{h})}{h}\ge0
\]
for any $(t,x)\in[0,T]\times R$. The equilibrium value function $V$
is defined as
\[
V(t,x)=J(t,x,\hat{\mathbf{u}}).
\]

\end{defn}

For a control law $\mathbf{u}$, we first define an infinitesimal
operator $\mathbf{A}^{\mathbf{u}}$:
\[
\mathbf{A}^{\mathbf{u}}=\frac{\partial}{\partial t}+\mu(x,\mathbf{u}(t,x))\frac{\partial}{\partial x}+\frac{1}{2}\sigma^{2}(x,\mathbf{u}(t,x))\frac{\partial^{2}}{\partial x^{2}}.
\]

\begin{defn}
\label{def:Extended-HJB-Equation}(Extended HJB Equation). For the
Nash equilibrium problem, the extended HJB system of equations for
$J$ is
\begin{align*}
\sup_{\mathbf{u}\in\mathcal{U}}\{(\mathbf{A}^{\mathbf{u}}V)(t,x)-(\mathbf{A}^{\mathbf{u}}f)(t,x,x)+(\mathbf{A}^{\mathbf{u}}f^{x})(t,x)\\
-\mathbf{A}^{\mathbf{u}}(G\diamond g)(t,x)+G_{y}(x,g(t,x))\cdot\mathbf{A}^{\mathbf{u}}g(t,x)\} & =0,\quad0\le t\le T,\\
\mathbf{A}^{\hat{\mathbf{u}}}f^{y}(t,x) & =0,\quad0\le t\le T,\\
\mathbf{A}^{\mathbf{\hat{\mathbf{u}}}}g(t,x) & =0,\quad0\le t\le T,\\
V(T,x) & =F(x,x)+G(x,x),\\
f(T,x,y) & =F(y,x),\\
g(T,x) & =x.
\end{align*}
In this HJB system, $\hat{\mathbf{u}}$ is the optimal control law
for the first equation. The notations $G_{y}$, $f^{y}$ and $G\diamond g$
are defined as
\begin{align*}
G_{y}(x,y) & =\frac{\partial}{\partial y}G(x,y),\\
f^{y}(t,x) & =f(t,x,y),\\
(G\diamond g)(t,x) & =G(x,g(t,x)).
\end{align*}
For functions $f$ and $g$, we have the following probabilistic interpretations:
\begin{align}
f(t,x,y) & =\mathrm{E}_{t,x}\big\{ F\big(y,X_{T}^{\hat{\mathbf{u}}}\big)\big\},\label{eq:Probabilistic_Interpretation_f}\\
g(t,x) & =\mathrm{\mathrm{E}}_{t,x}(X_{T}^{\hat{\mathbf{u}}}).\label{eq:Probabilistic_Interpretation_g}
\end{align}
\end{defn}
\begin{thm}
(Verification Theorem). Suppose that $(V,f,g)$ is a solution of the
HJB system defined in Definition \ref{def:Extended-HJB-Equation}
and the supremum in the first equation is attained at $\hat{\mathbf{u}}(t,x)$.
Then $\hat{\mathbf{u}}$ is an equilibrium control law and $V(t,x)$
is the corresponding value function. In addition, $f$ and $g$ allow
for the probabilistic interpretations \eqref{eq:Probabilistic_Interpretation_f}
and \eqref{eq:Probabilistic_Interpretation_g}.
\end{thm}
The proof of the Verification Theorem can be found in Section 3 of
Björk and Murgoci (2010). This theorem states that the solution of
the extended HJB system yields the optimal control and the value function
of the original stochastic control problem.

\section{Portfolio selection\label{sec:3_Models_Description}}

We consider the dynamics for a bank account $B$ and the prices of
risky stocks $S_{i}$:
\begin{align*}
dB_{t} & =rB_{t}dt\\
dS_{it} & =\alpha_{i}S_{it}dt+\sum_{j=1}^{d}\sigma_{ij}S_{it}dW_{jt},\quad i=1,\ldots,n,
\end{align*}
where $r$ is the risk-free rate of a bank account $B$, $\alpha_{i}$
is the mean return of stock $i$ and $\sigma_{ij}$ is the volatility
amount of stock $i$ affected by risk source $j$. Denote $u_{i}$
to be the dollar amount invested in the $i$-th stock and $\mathbf{u_{t}}=(u_{1t},\ldots,u_{nt})$
to be{\tiny{} }the corresponding vector. Two scenarios are considered
in this section: the first one is to include a risk-free asset and
multiple risky assets in our portfolio and the other one is to include
only multiple risky assets in our portfolio. With different objective
functions, we are going to derive the corresponding optimal dollar
amount $\hat{{\bf u}}_{t}$.

If the Brownian motions are correlated such that $\mathrm{E}(dW_{it}dW_{jt})=\rho_{ij}dt$
and $\rho_{ii}=1$, then the covariance between stock prices $S_{it}$
and $S_{jt}$ is
\[
\mathrm{Cov}(S_{it},S_{jt})=S_{i0}S_{j0}e^{(\alpha_{i}+\alpha_{j})t}\big(e^{\sum_{k=1}^{d}\tilde{\sigma}_{ik}\tilde{\sigma}_{jk}t}-1\big),
\]
where $\tilde{\sigma}_{i1}=\sum_{k=1}^{d}\rho_{1k}\sigma_{ik}$ and
$\tilde{\sigma}_{il}=\sqrt{1-\rho_{1l}^{2}}\sigma_{il}$ for $l=2,\ldots,d$.

For simplicity, in the models we are going to discuss, we consider
the case with $n=2$ and $d=2$, i.e.:
\begin{equation}
\begin{cases}
dS_{1t} & =\alpha_{1}S_{1t}dt+\sigma_{11}S_{1t}dW_{1t}+\sigma_{12}S_{1t}dW_{2t},\\
dS_{2t} & =\alpha_{2}S_{2t}dt+\sigma_{21}S_{2t}dW_{1t}+\sigma_{22}S_{2t}dW_{2t},
\end{cases}\label{eq:Stock_Prices_Dynamics}
\end{equation}
In addition, for two Brownian motions $W_{it}$ and $W_{jt}$ with
correlation coefficient $\rho_{ij}$, we can make the transformation
\begin{align*}
W_{it} & =\tilde{W}_{it},\\
W_{jt} & =\rho_{ij}\tilde{W}_{it}+\sqrt{1-\rho_{ij}^{2}}\tilde{W}_{jt},
\end{align*}
where $\mathrm{E}(d\tilde{W}_{it}d\tilde{W}_{jt})=0$. Therefore,
we can assume that the correlation coefficient between two different
Brownian motions is zero, i.e., $\rho_{ij}=0$ for $i\neq j$ and
$\rho_{ii}=1$.

\subsection{Revisit the mean--variance optimization with risk-free asset}

In this section, we revisit the model from Basak and Chabakauri (2010)
in a framework of complete market and do some special analysis which
will be easily compared with the next two models. In our model 1,
a risk-free asset and two risky assets are included in the portfolio.
The objective function is
\[
J(t,x,{\bf u})=\mathrm{E}_{t,x}\left(X_{T}^{\mathbf{u}}\right)-\frac{\gamma}{2}\mathrm{Var}_{t,x}\left(X_{T}^{\mathbf{u}}\right)
\]
where $\gamma\in R$ is the risk aversion coefficient. Denote $\boldsymbol{\sigma}_{i}=(\sigma_{i1},\sigma_{i2})^{\mathrm{T}}$
for $i=1,2$, $\boldsymbol{\sigma}=(\boldsymbol{\sigma}_{1},\boldsymbol{\sigma}_{2})^{\mathrm{T}}$
and $\mathbf{W}_{t}=(W_{1t},W_{2t})^{\mathrm{T}}$. Let $\mathbf{1}=(1,1)^{\mathrm{T}}$.
The dynamic of the investor's wealth is
\[
dX_{t}^{\mathbf{u}}=\big\{ rX_{t}^{\mathbf{u}}+(\boldsymbol{\alpha}-r\mathbf{1})^{\mathrm{T}}\mathbf{\mathbf{u}}_{t}\}dt+\mathbf{u}_{t}^{\mathrm{T}}\boldsymbol{\sigma}d\mathbf{W}_{t}
\]
where $\mathbf{u}_{t}=(u_{1t},u_{2t})^{\mathrm{T}}$ is the vector
of dollar investments in the two stocks at time $t$.

In this case, the functions $F(x)=x-\frac{\gamma}{2}x^{2}$ and $G(x)=\frac{\gamma}{2}x^{2}$.
The corresponding extended HJB equation is given by
\begin{align}
V_{t}+\sup_{\mathbf{u}\in R^{2}}\bigg[\{rx+(\boldsymbol{\alpha}-r\mathbf{1})^{\mathrm{T}}\mathbf{\mathbf{u}}\}V_{x}+\frac{1}{2}\big(V_{xx}-\gamma g_{x}^{2}\big)\mathbf{u}^{\mathrm{T}}\boldsymbol{\sigma}\boldsymbol{\sigma}^{\mathrm{T}}\mathbf{u}\bigg] & =0,\nonumber \\
g_{t}+\{rx+(\boldsymbol{\alpha}-r\mathbf{1})^{\mathrm{T}}\mathbf{\mathbf{u}}\}g_{x}+\frac{1}{2}g_{xx}\mathbf{u}^{\mathrm{T}}\boldsymbol{\sigma}\boldsymbol{\sigma}^{\mathrm{T}}\mathbf{u} & =0,\label{eq:HJB_Model_1}\\
V(T,x) & =x,\nonumber \\
g\left(T,x\right) & =x.\nonumber
\end{align}

Assuming the $2\times2$ matrix $\boldsymbol{\sigma}$ is invertible,
By Section 3.2 in Basak and Chabakauri (2010), we have the solutions
for equation \eqref{eq:HJB_Model_1}:
\begin{align*}
V(t,x) & =e^{r(T-t)}x+\frac{1}{2\gamma}(\boldsymbol{\alpha}-r\mathbf{1})^{\mathrm{T}}(\boldsymbol{\sigma}\boldsymbol{\sigma}^{\mathrm{T}})^{-1}(\boldsymbol{\alpha}-r\mathbf{1})(T-t),\\
g(t,x) & =e^{r(T-t)}x+\frac{1}{\gamma}(\boldsymbol{\alpha}-r\mathbf{1})^{\mathrm{T}}(\boldsymbol{\sigma}\boldsymbol{\sigma}^{\mathrm{T}})^{-1}(\boldsymbol{\alpha}-r\mathbf{1})(T-t),
\end{align*}
and the equilibrium control is
\[
\hat{\mathbf{u}}(t,x)=\frac{1}{\gamma}e^{-r(T-t)}(\boldsymbol{\sigma}\boldsymbol{\sigma}^{\mathrm{T}})^{-1}(\boldsymbol{\alpha}-r\mathbf{1}).
\]

The numbers of risk sources and risky assets affect the equilibrium
drastically, as we will see in the following analysis.

\subsubsection*{One Brownian motion case: }

Suppose that there is only one random factor affecting the market.
Without loss of generality, let $\sigma_{12}=\sigma_{22}=0$ in \eqref{eq:Stock_Prices_Dynamics}.
The wealth process can be written as
\begin{equation}
dX_{t}^{\mathbf{u}}=\{X_{t}^{\mathrm{\mathbf{u}}}r+(\alpha_{1}-r)u_{1}+(\alpha_{2}-r)u_{2}\}dt+(\sigma_{11}u_{1}+\sigma_{21}u_{2})dW_{t}\label{eq:wealth_process_model_1}
\end{equation}
where $W_{t}$ is a 1-dim Brownian motion. \textcolor{black}{From
the first equation of \eqref{eq:HJB_Model_1} and the first order
condition, we have
\[
(\boldsymbol{\sigma}^{\mathrm{T}}\boldsymbol{\sigma})\mathbf{u}=\frac{V_{x}}{\gamma g_{x}^{2}-V_{xx}}(\mathbf{\boldsymbol{\alpha}}-r\mathbf{1}),
\]
or precisely,
\begin{equation}
\begin{cases}
\sigma_{11}u_{1}+\sigma_{21}u_{2} & =\frac{\alpha_{1}-r}{\sigma_{11}}\frac{V_{x}}{\gamma g_{x}^{2}-V_{xx}},\\
\sigma_{11}u_{1}+\sigma_{21}u_{2} & =\frac{\alpha_{2}-r}{\sigma_{21}}\frac{V_{x}}{\gamma g_{x}^{2}-V_{xx}}.
\end{cases}\label{eq:u1_u2_1BM_Model1}
\end{equation}
}
\begin{itemize}
\item[(i)] \textcolor{black}{If the market prices of risk of two stocks are
equal, i.e., $(\alpha_{1}-r)/\sigma_{11}=(\alpha_{2}-r)/\sigma_{21}$,
the optimal amounts of money $\hat{u}_{1}$ and $\hat{u}_{2}$ only
satisfy
\[
\sigma_{11}\hat{u}_{1}+\sigma_{21}\hat{u}_{2}=\frac{\alpha_{1}-r}{\sigma_{11}}\frac{V_{x}}{\gamma g_{x}^{2}-V_{xx}}
\]
which are not unique. With this relationship, we can obtain the solution
for \eqref{eq:HJB_Model_1}:
\begin{align*}
V(t,x) & =e^{r(T-t)}x+\frac{1}{2\gamma}\big(\frac{\alpha_{1}-r}{\sigma_{1}}\big)^{2}(T-t),\\
g(t,x) & =e^{r(T-t)}x+\frac{1}{\gamma}\big(\frac{\alpha_{1}-r}{\sigma_{1}}\big)^{2}(T-t).
\end{align*}
The corresponding linear combination of the optimal amounts of money
is thus
\[
\hat{u}_{1}\sigma_{11}+\hat{u}_{2}\sigma_{21}=\frac{1}{\gamma}\frac{\alpha_{1}-r}{\sigma_{11}}e^{-r(T-t)}.
\]
This is consistent with financial intuition, because with the same
price of market risk, it does not matter to buy one of them more or
less.}
\item[(ii)] \textcolor{black}{If $(\alpha_{1}-r)/\sigma_{11}\neq(\alpha_{2}-r)/\sigma_{21}$,
since $V_{x}\neq0$, there exists no solution for \eqref{eq:u1_u2_1BM_Model1}.}\footnote{\textcolor{black}{If $V_{x}\equiv0$, the first equation of the HJB
system \eqref{eq:HJB_Model_1} becomes $V_{t}\equiv0$ and thus $V(t,x)$
is a constant for all $t\in[0,T]$ and $x\in R$. This contradicts
with the boundary condition $V(T,x)=x$. }}\textcolor{black}{{} We have two prices of market risk. It is clear
to see the case when $\sigma_{11}=\sigma_{21}$ but $\alpha_{1}>\alpha_{2}$.
This obviously implies arbitrage. Intuitively in this circumstance,
one should buy the stock with higher price of market risk, i.e., buy
stock 1 as many as he can. }
\end{itemize}

\subsubsection*{Two Brownian motions case: }

If we still have two risky assets but the uncertainties of the prices
of these two assets are decided by two independent Brownian motions,
i.e., $\sigma_{12}=\sigma_{21}=0$ in \eqref{eq:Stock_Prices_Dynamics}.
Without loss generosity, we further assume that $\sigma_{11}=\sigma_{22}=\sigma$,
i.e., the volatilities of two stocks are the same.
\begin{itemize}
\item[(i)]  If $\alpha_{1}=\alpha_{2}=\alpha$, then
\[
\hat{u}_{1}(t,x)=\hat{u}_{2}(t,x)=\frac{1}{\gamma}e^{-r(T-t)}\frac{\alpha-r}{\sigma^{2}},
\]
i.e., they are the same with the optimal one derived for the situation
when only one stock is available. Note that here the two stocks are
not exactly the same with each other because they are randomized by
two independent Brownian motions, which demonstrates that an investor's
decision will be affected by the appearance parameters $\alpha$ and
$\sigma$.
\item[(ii)]  If $\alpha_{1}\neq\alpha_{2}$, then
\[
\hat{u}_{i}(t,x)=\frac{1}{\gamma}e^{-r(T-t)}\frac{\alpha_{i}-r}{\sigma^{2}},\qquad i=1,2,
\]
i.e., if the two stocks have the same volatility, the money invested
on each stock is positive-related to its appreciation rate.
\end{itemize}

\subsection{The variance model with only two stocks}

In this model, we consider that a risk-loving investor who will not
put money into bank account and that the variance as the objective
function
\[
J(t,x,u)=-\frac{\gamma}{2}\mathrm{Var}_{t,x}(X_{T}^{u}).
\]
Assume there are only two stocks available, we denote the amount of
money invested in stock 1 at time $t$ to be $u_{t}$. The amount
of money invested in stock 2 is thus $X_{t}-u_{t}$ and the dynamic
of the value process of the portfolio is\textcolor{black}{
\begin{equation}
dX_{t}^{u}=\left\{ u_{t}\alpha_{1}+(X_{t}^{u}-u_{t})\alpha_{2}\right\} dt+\left\{ u_{t}\sigma_{11}+(X_{t}^{u}-u_{t})\sigma_{21}\right\} dW_{1t}+\left\{ u_{t}\sigma_{12}+(X_{t}^{u}-u_{t})\sigma_{22}\right\} dW_{2t}.\label{eq:Wealth_Process_Model_2_3}
\end{equation}
}The corresponding extended HJB equation is given by
\begin{equation}
\begin{aligned}V_{t}+\sup_{u\in\mathcal{U}}\bigg(\left\{ x\alpha_{2}+(\alpha_{1}-\alpha_{2})u\right\} V_{x}+\frac{1}{2}\bigg[\left\{ x\sigma_{21}+\left(\sigma_{11}-\sigma_{21}\right)u\right\} ^{2}\\
+\left\{ x\sigma_{22}+\left(\sigma_{12}-\sigma_{22}\right)u\right\} ^{2}\bigg]\left(V_{xx}-\gamma g_{x}^{2}\right)\bigg) & =0,\\
g_{t}+\left\{ x\alpha_{2}+\left(\alpha_{1}-\alpha_{2}\right)u\right\} g_{x}+\frac{1}{2}\bigg[\left\{ x\sigma_{21}+\left(\sigma_{11}-\sigma_{21}\right)u\right\} ^{2}\\
+\left\{ x\sigma_{22}+\left(\sigma_{12}-\sigma_{22}\right)u\right\} ^{2}\bigg]g_{xx} & =0,\\
V(T,x) & =0,\\
g\left(T,x\right) & =x.
\end{aligned}
\label{eq:HJB_Model_2}
\end{equation}
For the optimal solution $\hat{u}$, we make the \textit{Ansatz:}\footnote{As to why make this \textit{Ansatz}, we have first tried the form
$\hat{u}(t,x)=k(t)x+c(t)$ which leads to a contradiction when putting
$g(t,x)=a(t)x$, $V(t,x)=A(t)x$ into the HJB system \eqref{eq:HJB_Model_2}.}
\begin{equation}
\hat{u}(t,x)=k(t)x.\label{eq:optimal_allocation_Model_2}
\end{equation}
By substituting equation (\ref{eq:optimal_allocation_Model_2}) into
equation (\ref{eq:Wealth_Process_Model_2_3}), we obtain the  wealth
process:
\[
dX_{t}^{\hat{u}}=\{\alpha_{2}+(\alpha_{1}-\alpha_{2})k(t)\}X_{t}^{\hat{u}}dt+\{\sigma_{21}+(\sigma_{11}-\sigma_{21})k(t)\}X_{t}^{\hat{u}}dW_{1t}+\{\sigma_{22}+(\sigma_{12}-\sigma_{22})k(t)\}X_{t}^{\hat{u}}dW_{2t}.
\]
From this equation, we can obtain the expected values:
\begin{align*}
\mathrm{E}_{t,x}\big(X_{T}^{\hat{u}}\big) & =e^{\int_{t}^{T}\left\{ \alpha_{2}+k\left(s\right)\left(\alpha_{1}-\alpha_{2}\right)\right\} ds}x,\\
\mathrm{E}_{t,x}\big\{\big(X_{T}^{\hat{u}}\big)^{2}\big\} & =e^{2\int_{t}^{T}\big[\alpha_{2}+k\left(s\right)\left(\alpha_{1}-\alpha_{2}\right)+0.5\left\{ \sigma_{21}+k(s)(\sigma_{11}-\sigma_{21})\right\} ^{2}+0.5\left\{ \sigma_{22}+k(s)(\sigma_{12}-\sigma_{22})\right\} ^{2}\big]ds}x^{2}.
\end{align*}
Therefore, the conditional variance of wealth is
\begin{align*}
\mathrm{Var}_{t,x}\big(X_{T}^{\hat{u}}\big) & =\mathrm{E}_{t,x}\big\{\big(X_{T}^{\hat{u}}\big)^{2}\big\}-\big\{\mathrm{E}_{t,x}\big(X_{T}^{\hat{u}}\big)\big\}^{2}\\
 & =e^{2\int_{t}^{T}\left\{ \alpha_{2}+k\left(s\right)\left(\alpha_{1}-\alpha_{2}\right)\right\} ds}\left(e^{\intop_{t}^{T}\left[\left\{ \sigma_{21}+k(s)(\sigma_{11}-\sigma_{21})\right\} ^{2}+\left\{ \sigma_{22}+k(s)(\sigma_{12}-\sigma_{22})\right\} ^{2}\right]ds}-1\right)x^{2}.
\end{align*}
The solution for equation \eqref{eq:HJB_Model_2} is given by
\begin{align*}
g(t,x) & =\mathrm{E}_{t,x}\big(X_{T}^{\hat{u}}\big)\\
 & =e^{\int_{t}^{T}\left\{ \alpha_{2}+k\left(s\right)\left(\alpha_{1}-\alpha_{2}\right)\right\} ds}x,\\
V(t,x) & =-\frac{\gamma}{2}\mathrm{Var}_{t,x}\big(X_{T}^{\hat{u}}\big)\\
 & =\frac{\gamma}{2}e^{2\intop_{t}^{T}\left\{ \alpha_{2}+k\left(s\right)\left(\alpha_{1}-\alpha_{2}\right)\right\} ds}\bigg(1-e^{\intop_{t}^{T}\left[\left\{ \sigma_{21}+k(s)(\sigma_{11}-\sigma_{21})\right\} ^{2}+\left\{ \sigma_{22}+k(s)(\sigma_{12}-\sigma_{22})\right\} ^{2}\right]ds}\bigg)x^{2}.
\end{align*}
Denote
\begin{align*}
a(t) & =e^{\int_{t}^{T}\left\{ \alpha_{2+}k\left(s\right)\left(\alpha_{1}-\alpha_{2}\right)\right\} ds},\\
A(t) & =\frac{\gamma}{2}e^{2\intop_{t}^{T}\left\{ \alpha_{2}+k\left(s\right)\left(\alpha_{1}-\alpha_{2}\right)\right\} ds}\bigg(1-e^{\intop_{t}^{T}\left[\left\{ \sigma_{21}+k(s)(\sigma_{11}-\sigma_{21})\right\} ^{2}+\left\{ \sigma_{22}+k(s)(\sigma_{12}-\sigma_{22})\right\} ^{2}\right]ds}\bigg),
\end{align*}
then $V(t,x)=A(t)x^{2}$ and $g(t,x)=a(t)x$.
\begin{thm}
Assume the volatilities of the two stocks are not exactly the same,
i.e., $(\sigma_{11}-\sigma_{21})^{2}+(\sigma_{12}-\sigma_{22})^{2}\ne0$,
from HJB equation \eqref{eq:HJB_Model_2} and the first order condition,
we obtain the expression of the optimal allocation
\begin{align*}
\hat{u}(t,x) & =\frac{(\alpha_{2}-\alpha_{1})V_{x}-\{\sigma_{21}(\sigma_{11}-\sigma_{21})+\sigma_{22}(\sigma_{12}-\sigma_{22})\}(V_{xx}-\gamma g_{x}^{2})x}{\{(\sigma_{11}-\sigma_{21})^{2}+(\sigma_{12}-\sigma_{22})^{2}\}(V_{xx}-\gamma g_{x}^{2})}\\
 & =\frac{2(\alpha_{2}-\alpha_{1})A(t)-\{\sigma_{21}(\sigma_{11}-\sigma_{21})+\sigma_{22}(\sigma_{12}-\sigma_{22})\}\big\{2A(t)-\gamma a^{2}(t)\big\}}{\{(\sigma_{11}-\sigma_{21})^{2}+(\sigma_{12}-\sigma_{22})^{2}\}\big\{2A(t)-\gamma a^{2}(t)\big\}}x\\
 & =k(t)x,
\end{align*}
where $k(\cdot)$ satisfies the following ordinary differential equation
(ODE):\footnote{The proof of the uniqueness and existence of a solution to \eqref{eq:k_for_model_2}
is the same as that for the ODE of $k_{1}(t)$, which is given in
Appendix \ref{sec: Appendix_B}. }
\begin{align}
k(t) & =\frac{1}{\left(\sigma_{11}-\sigma_{21}\right)^{2}+\left(\sigma_{12}-\sigma_{22}\right)^{2}}\Bigg\{\left(\alpha_{1}-\alpha_{2}\right)\bigg(e^{-\intop_{t}^{T}\left[\left\{ \sigma_{21}+k(s)(\sigma_{11}-\sigma_{21})\right\} ^{2}+\left\{ \sigma_{22}+k(s)(\sigma_{12}-\sigma_{22})\right\} ^{2}\right]ds}\nonumber \\
 & \;\;\;\;-1\bigg)-\sigma_{21}\left(\sigma_{11}-\sigma_{21}\right)-\sigma_{22}\left(\sigma_{12}-\sigma_{22}\right)\Bigg\}.\label{eq:k_for_model_2}
\end{align}

\end{thm}
In the case where the first asset is a stock with $\alpha_{1}=\alpha$,
$\sigma_{11}=\sigma>0$, $\sigma_{12}=0$ and the second asset degenerates
to risk-free asset with $\alpha_{2}=r$, $\sigma_{21}=\sigma_{22}=0$,
this equation has a unique solution $k(t)\equiv0$ for all $t\in[0,T]$,
i.e., we do not invest any money into the risky asset. This is reasonable
because in this model, the only consideration for an investor to construct
his portfolio is to minimize his risk, therefore he invests all of
his wealth into the risk-free asset in order to avoid taking any risk.

\subsubsection*{One Brownian motion case}

\textcolor{black}{Suppose that there is only one random factor affecting
the market. Without loss of generality, let $\sigma_{12}=\sigma_{22}=0$.
By applying the first order condition to the first equation of the
HJB system \eqref{eq:HJB_Model_2}, we have:
\begin{equation}
\{(\sigma_{11}-\sigma_{21})^{2}+(\sigma_{12}-\sigma_{22})^{2}\}u=(\alpha_{1}-\alpha_{2})\frac{V_{x}}{\gamma g_{x}^{2}-V_{xx}}-\{\sigma_{21}(\sigma_{11}-\sigma_{21})+\sigma_{22}(\sigma_{12}-\sigma_{22})\}x.\label{eq:u_1BM_Model_2_3}
\end{equation}
}
\begin{itemize}
\item[(i)] \textcolor{black}{ If $\sigma_{11}=\sigma_{21}$ and $\alpha_{1}=\alpha_{2}$,
\eqref{eq:u_1BM_Model_2_3} is always true no matter what real number
$u$ takes as the two sides of this equation will always be 0. In
fact, these two stocks are ``the same''. }
\item[(ii)] \textcolor{black}{ If $\sigma_{11}=\sigma_{21}$ but $\alpha_{1}\neq\alpha_{2}$,
since $V_{x}\neq0$, no solution exists for \eqref{eq:u_1BM_Model_2_3}.}\footnote{\textcolor{black}{By substituting $V_{x}\equiv0$ into the first equation
of the HJB system \eqref{eq:HJB_Model_2}, we obtain $V_{t}=\frac{1}{2}\gamma\sigma_{11}^{2}x^{2}g_{x}^{2}$
for any function $u$. Therefore, we can set $u=0$ in the second
equation of \eqref{eq:HJB_Model_2} and the corresponding solution
is $g(t,x)=xe^{\alpha_{2}(T-t)}$. With this, we have $V_{t}=\frac{1}{2}\gamma\sigma_{11}^{2}e^{2\alpha_{2}(T-t)}x^{2}$
which leads to a contradiction with $V_{x}\equiv0$. }}
\end{itemize}

\subsubsection*{Two Brownian motions case:}

If uncertainties of the prices of the two assets are decided by two
independent Brownian motions with the same volatility amount, i.e.,
$\sigma_{11}=\sigma_{22}=\sigma$ and $\sigma_{12}=\sigma_{21}=0$,
the optimal allocation is $\hat{u}(t,x)=k(t)x$ where
\[
k(t)=\frac{\alpha_{1}-\alpha_{2}}{2\sigma_{11}^{2}}\big[e^{-\int_{t}^{T}\sigma_{11}^{2}\{2k^{2}(s)-2k(s)+1\}ds}-1\big]+\frac{1}{2}.
\]
At time $t\in[0,T]$, the expectation of the wealth at the end of
the time period is given by
\begin{equation}
\mathrm{E}_{t,x}(X_{T}^{u})=xe^{\int_{t}^{T}\left\{ \alpha_{2}+k\left(s\right)\left(\alpha_{1}-\alpha_{2}\right)\right\} ds},\label{eq:Exp_Scenario2_model_2}
\end{equation}
while its variance is
\begin{equation}
\mathrm{Var}_{t,x}(X_{T}^{u})=x^{2}e^{2\int_{t}^{T}\left\{ \alpha_{2}+k\left(s\right)\left(\alpha_{1}-\alpha_{2}\right)\right\} ds}\big[e^{\int_{t}^{T}\sigma_{11}^{2}\{2k^{2}(s)-2k(s)+1\}ds}-1\big].\label{eq:Var_Scenario2_model_2}
\end{equation}

\begin{itemize}
\item[(i)]  If $\alpha_{1}=\alpha_{2}=\alpha$, then
\[
\hat{u}_{1}(t,x)=\hat{u}_{2}(t,x)=\frac{1}{2}x,
\]
i.e., the amounts of money invested on two stocks are both a half
at any time $t$ and without short-selling. This is reasonable from
the point of finance: there is no rank between the two random factors
(Brownian motions), and the two stocks perform at the same level ($\sigma_{11}=\sigma_{22}$,
$\alpha_{1}=\alpha_{2}$). So there is no reason to put more emphasis
on one stock. But things become different for the next two cases.
\item[(ii)]  if $\alpha_{1}>\alpha_{2}$, then
\[
\hat{u}_{1}(t,x)<\frac{1}{2}x<\hat{u}_{2}(t,x).
\]

\item[(iii)]  if $\alpha_{1}<\alpha_{2}$, then
\[
\hat{u}_{1}(t,x)>\frac{1}{2}x>\hat{u}_{2}(t,x).
\]

\end{itemize}
From (ii) and (iii), we can see that under the same volatility, the
amount of money invested into the stock with higher appreciation rate
is less than that invested into the stock with lower appreciation
rate. This leads to no contradiction under the criteria of variance
because a large $\alpha$ yields a large variance.

\subsection{Mean--variance criteria without bank account}

For the third model, as in model 2, the portfolio only includes two
risky assets. However, the objective function is the same with model
1, i.e.,
\[
J(t,x,u)=\mathrm{E}_{t,x}(X_{T}^{u})-\frac{\gamma}{2}\mathrm{Var}_{t,x}(X_{T}^{u}).
\]
Given the amount of money invested in stock 1 $u(t,x)$, the dynamic
of wealth is the same with the one in model 2, i.e., equation (\ref{eq:Wealth_Process_Model_2_3}).
The corresponding extended HJB equation is given by
\begin{equation}
\begin{aligned}V_{t}+\sup_{u\in\mathcal{U}}\bigg(\left\{ x\alpha_{2}+u\left(\alpha_{1}-\alpha_{2}\right)\right\} V_{x}+\frac{1}{2}\bigg[\left\{ x\sigma_{21}+u\left(\sigma_{11}-\sigma_{21}\right)\right\} ^{2}\\
+\left\{ x\sigma_{22}+u\left(\sigma_{12}-\sigma_{22}\right)\right\} ^{2}\bigg]\left(V_{xx}-\gamma g_{x}^{2}\right)\bigg) & =0,\\
g_{t}+\left\{ x\alpha_{2}+u\left(\alpha_{1}-\alpha_{2}\right)\right\} g_{x}+\frac{1}{2}\bigg[\left\{ x\sigma_{21}+u\left(\sigma_{11}-\sigma_{21}\right)\right\} ^{2}\\
+\left\{ x\sigma_{22}+u\left(\sigma_{12}-\sigma_{22}\right)\right\} ^{2}\bigg]g_{xx} & =0,\\
V(T,x) & =x,\\
g\left(T,x\right) & =x.
\end{aligned}
\label{eq:HJB_Model_3}
\end{equation}
The only difference between \eqref{eq:HJB_Model_2} and \eqref{eq:HJB_Model_3}
is that $V(T,x)$ takes different value. Suppose the optimal allocation
for this problem is in the form of \footnote{The procedure for solving this HJB system heavily depends on the \textit{Ansatz}
of $\hat{u}(t,x)$. We tried this general linear form and as we will
see later, neither $k_{1}$ nor $k_{2}$ equals to zero in this model.}
\begin{equation}
\hat{u}(t,x)=k_{1}(t)x+k_{2}(t).\label{eq:optimal_allocation_Model_3}
\end{equation}

\begin{thm}
Assume that $(\sigma_{11}-\sigma_{21})^{2}+(\sigma_{12}-\sigma_{22})\ne0$.
The optimal allocation is $\hat{u}(t,x)=k_{1}(t)x+k_{2}(t)$ where
$k_{1}$ and $k_{2}$ satisfy the ODE system: \label{thm:optimal_allocation_u_model_3}
\begin{align}
k_{1}(t) & =\frac{1}{\left(\sigma_{11}-\sigma_{21}\right)^{2}+\left(\sigma_{12}-\sigma_{22}\right)^{2}}\Bigg\{\left(\alpha_{1}-\alpha_{2}\right)\left(e^{-\intop_{t}^{T}\left[\left\{ \sigma_{21}+k_{1}(s)(\sigma_{11}-\sigma_{21})\right\} ^{2}+\left\{ \sigma_{22}+k_{1}(s)(\sigma_{12}-\sigma_{22})\right\} ^{2}\right]ds}-1\right)\nonumber \\
 & \;\;\;\;-\sigma_{21}\left(\sigma_{11}-\sigma_{21}\right)-\sigma_{22}\left(\sigma_{12}-\sigma_{22}\right)\Bigg\},\label{eq:k1_model_3}\\
k_{2}(t) & =\frac{\alpha_{1}-\alpha_{2}}{\left(\sigma_{11}-\sigma_{21}\right)^{2}+\left(\sigma_{12}-\sigma_{22}\right)^{2}}\bigg\{\frac{1}{\gamma}I_{1}(t,T)I_{2}(t,T)+\int_{t}^{T}I_{1}(t,v)I_{3}(t,v)k_{2}(v)dv\bigg\},\label{eq:k2_model_3}
\end{align}
where
\begin{align*}
I_{1}(t,v) & =e^{-\int_{t}^{v}\{\alpha_{2}+k_{1}(s)(\alpha_{1}-\alpha_{2})\}ds},\\
I_{2}(t,v) & =e^{-\int_{t}^{v}\big[\{\sigma_{21}+k_{1}(s)(\sigma_{11}-\sigma_{21})\}^{2}+\{\sigma_{12}+k_{1}(s)(\sigma_{12}-\sigma_{22})\}^{2}\big]ds},\\
I_{3}(t,v) & =(\alpha_{1}-\alpha_{2})I_{2}(t,T)-[(\alpha_{1}-\alpha_{2})+(\sigma_{11}-\sigma_{21})\{\sigma_{21}+k_{1}(v)(\sigma_{11}-\sigma_{21})\}\\
 & \;\;\;\;+(\sigma_{12}-\sigma_{22})\{\sigma_{22}+k_{1}(v)(\sigma_{12}-\sigma_{22})\}]I_{2}(t,v).
\end{align*}
\end{thm}
\begin{proof}
The proof of Theorem \ref{thm:optimal_allocation_u_model_3} is quite
tedious. We put it in Appendix \ref{sec: Appendix_A}. \end{proof}
\begin{prop}
The ODE system \eqref{eq:k1_model_3} and \eqref{eq:k2_model_3} admits
a unique solution $(k_{1}(t),k_{2}(t))^{\mathrm{T}}$ where $k_{1},k_{2}\in C[0,T]$.\label{thm:ODE_Unique_Solution}\end{prop}
\begin{proof}
See Appendix \ref{sec: Appendix_B}.
\end{proof}
We now make some comments and analysis.
\begin{itemize}
\item Different from the equilibrium control law $\hat{u}(t,x)=c(t)x$ in
Björk et al. (2014) where $c(t)$ is a function of $\gamma$, the
corresponding part which associates with the current wealth state,
$k_{1}(t)$, in the equilibrium control law of model 3 does not depend
on $\gamma$.
\item Suppose that $\left(\sigma_{11}-\sigma_{21}\right)^{2}+\left(\sigma_{12}-\sigma_{22}\right)^{2}$
is strictly positive,

\begin{enumerate}
\item If $\alpha_{1}>(<)\alpha_{2}$, the values of $k_{2}(t)$ and $\hat{u}(t,x)$
decrease (increase) when $\gamma$ increases.
\item If $\alpha_{1}=\alpha_{2}$, then $\forall t\in[0,T]$, we have
\begin{align}
k_{1}(t) & \equiv-\frac{1}{\left(\sigma_{11}-\sigma_{21}\right)^{2}+\left(\sigma_{12}-\sigma_{22}\right)^{2}}\Bigg\{\sigma_{21}\left(\sigma_{11}-\sigma_{21}\right)+\sigma_{22}\left(\sigma_{12}-\sigma_{22}\right)\Bigg\},\label{eq:k1_model_3_constant_alpha1_equals_alpha_2}\\
k_{2}(t) & \equiv0,\nonumber
\end{align}
i.e., $k_{1}(t)$ and $k_{2}(t)$ are constants with respect to $t$.
In this case, $\hat{u}(t,x)=k_{1}(t)x$ is directly proportional to
$x$ with a proportionality constant that does not change with time
$t$. If we further assume $\sigma_{12}=\sigma_{22}=0$, the constant
$k_{1}(t)$ is larger than one and thus $\hat{u}(t,x)>x$ when $\sigma_{11}<\sigma_{21}$.
In this case, we long stock 1 and short stock 2. On the other hand,
if $\sigma_{11}>\sigma_{21}$, we have $\hat{u}(t,x)<0$ and this
indicates that we short stock 1 and long stock 2.
\end{enumerate}
\item Furthermore, by differentiating $k_{1}(t)$ in equation \eqref{eq:k1_model_3}
with respect to $t$, we have
\begin{align*}
\frac{d}{dt}k_{1}(t) & =\frac{\alpha_{1}-\alpha_{2}}{\left(\sigma_{11}-\sigma_{21}\right)^{2}+\left(\sigma_{12}-\sigma_{22}\right)^{2}}e^{-\intop_{t}^{T}\left[\left\{ \sigma_{21}+k_{1}(s)(\sigma_{11}-\sigma_{21})\right\} ^{2}+\left\{ \sigma_{22}+k_{1}(s)(\sigma_{12}-\sigma_{22})\right\} ^{2}\right]ds}\\
 & \;\;\;\;\times\big[\left\{ \sigma_{21}+k_{1}(t)(\sigma_{11}-\sigma_{21})\right\} ^{2}+\left\{ \sigma_{22}+k_{1}(t)(\sigma_{12}-\sigma_{22})\right\} ^{2}\big].
\end{align*}
Therefore, if $\alpha_{1}>(<)\alpha_{2}$, then $dk_{1}(t)/dt>(<)0$
and thus the value of $k_{1}(t)$ increases (decreases) as $t$ increases;
\item Notice that for this model, if we let $\sigma_{12}=\sigma_{21}=\sigma_{22}=0$,
then the uncertainty of the price of the first stock is controlled
by only one Brownian motion and the second stock becomes a riskless
asset. In such case, model 3 is identical to the one considered in
Section 6.1 of Björk and Murgoci (2010). The optimal allocation $\hat{u}(t,x)$,
the expected value of the optimal portfolio $\mathrm{E}_{t,x}(X_{T}^{\hat{u}})$
and the equilibrium value function $V(t,x)$ obtained using equation
(\ref{eq:optimal_allocation_Model_3}), (\ref{eq:Exp_Model_3}) and
(\ref{eq:Value_Function_Model_3}) coincide with those in Section
6.1 of Björk and Murgoci (2010).
\end{itemize}

\subsubsection*{\textcolor{black}{One Brownian motion case: }}

\textcolor{black}{For the case where there is only one Brownian motion,
i.e., $\sigma_{12}=\sigma_{22}=0$ and $\sigma_{11}=\sigma_{21}>0$.
Since the first equation of the HJB system \eqref{eq:HJB_Model_3}
is the same with the one in \eqref{eq:HJB_Model_2}. Therefore, from
the first order condition, we have \eqref{eq:u_1BM_Model_2_3}. }
\begin{itemize}
\item[(i)] \textcolor{black}{ If $\alpha_{1}=\alpha_{2}$, we have the same
conclusion with the one made in model 2. }
\item[(ii)] \textcolor{black}{ If $\alpha_{1}\ne\alpha_{2}$, since $V_{x}\neq0$
with the similar deduction in model 2, therefore, no optimal solution
exists for \eqref{eq:HJB_Model_3}. }
\end{itemize}

\subsubsection*{\textcolor{black}{Two Brownian motions case: }}

\textcolor{black}{Suppose there are two independent Brownian motions,
i.e., $\sigma_{12}=\sigma_{21}=0$. In addition, we assume $\sigma_{11}=\sigma_{22}>0$. }
\begin{itemize}
\item[(i)] \textcolor{black}{ If $\alpha_{1}=\alpha_{2}=\alpha$, then
\[
\hat{u}(t,x)=\frac{1}{2}x,
\]
i.e., the amount of money invested on each stock is the same at any
time $t$ which coincides with the same case in model 2. With this
optimal allocation, the expectation and variance of the wealth $X_{T}$
are (\ref{eq:Exp_Scenario2_model_2}) and (\ref{eq:Var_Scenario2_model_2}).
The reward function can be written accordingly. }
\item[(ii)] \textcolor{black}{ If $\alpha_{1}\neq\alpha_{2}$, then no conclusion
can be made yet since the explicit solution of $\hat{u}(t,x)=k_{1}(t)x+k_{2}(t)$
cannot be obtained. Numerical analysis is required for studying the
behavior of $k_{1}(t,x)$ and $k_{2}(t,x)$. As a demonstration here,
we calculate $k_{1}(t)$ and $k_{2}(t)$ for several combinations
of $T$ and $\gamma$. The parameters are given by: $\alpha_{1}=0.2$,
$\alpha_{2}=0.12$, $\sigma_{11}=\sigma_{22}=0.25$, $\sigma_{12}=\sigma_{21}=0$.
The results are shown in Figure \ref{fig: k_1_k_2_model_3_2BM_demonstration}.
From }\textbf{\textcolor{black}{Figure \ref{fig: k_1_k_2_model_3_2BM_demonstration}}}\textcolor{black}{,
we have the following observations: }
\begin{figure}[h]
\begin{centering}
\textcolor{black}{\includegraphics[scale=0.4]{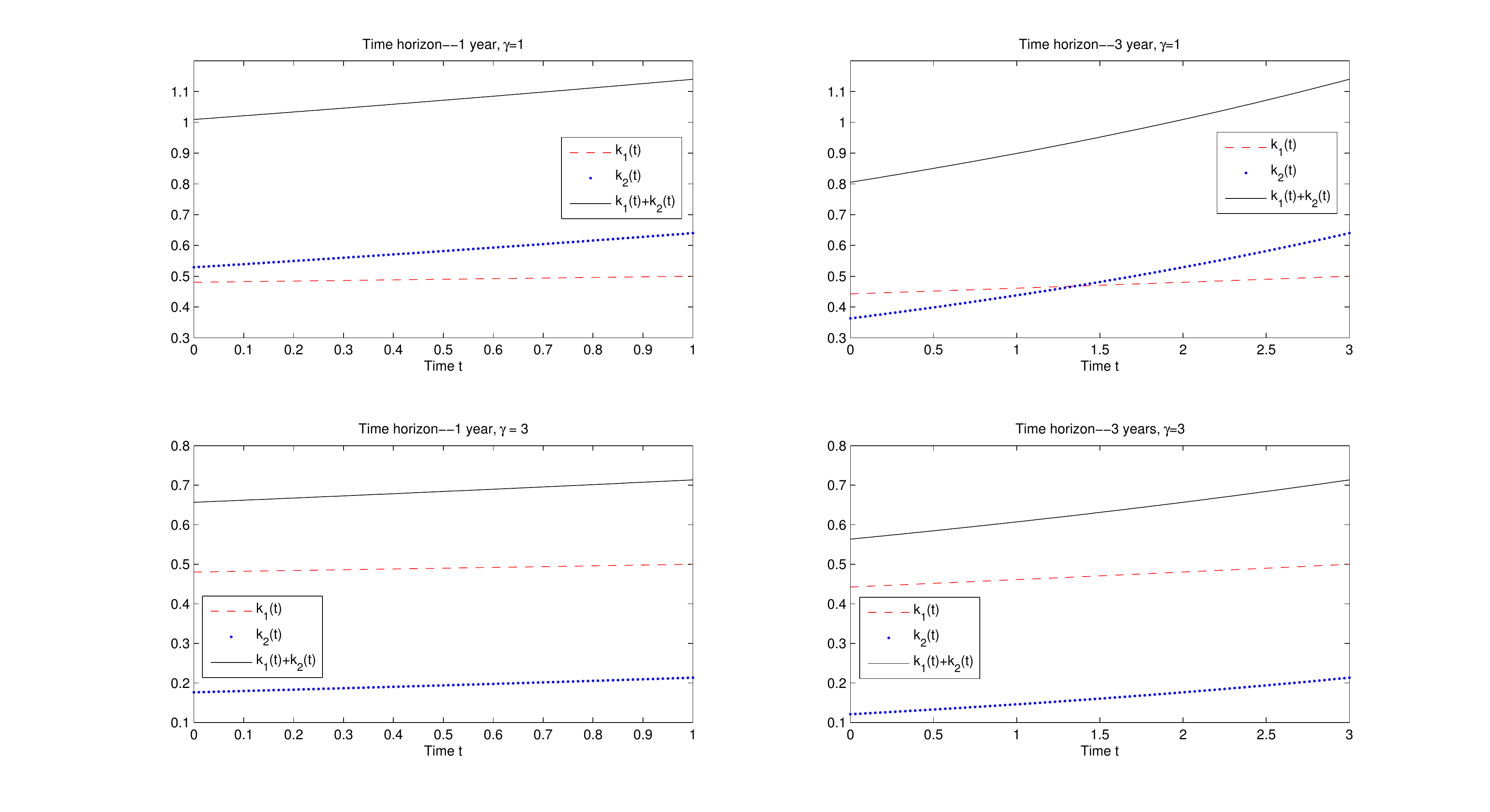}}
\par\end{centering}

\textcolor{black}{\caption{The functions $k_{1}(t)$ and $k_{2}(t)$ for various choices of $T$
and $\gamma$.\label{fig: k_1_k_2_model_3_2BM_demonstration}}
}
\end{figure}

\begin{itemize}
\item \textcolor{black}{For $k_{1}(t)$, we have the same conclusion as
the one obtained in model 2, i.e., $k_{1}(t)<0.5$ if $\alpha_{1}>\alpha_{2}$.
However, since the first stock has a higher appreciation rate, when
taking the expected wealth at time $T$ into consideration in our
objective function, an additional positive amount of money $k_{2}(t)$
($k_{2}(t)\equiv0$ in model 2) is required to invest on stock 1 at
any time $t\in[0,T]$. This can be viewed as a tradeoff between maximizing
the expected wealth and minimizing the corresponding variance.}
\item \textcolor{black}{As $\gamma$ increases from 1 to 3, the value of
$k_{2}(t)$ decreases drastically. This is expected as $\gamma$ increases,
the solution of model 3 will converge to that of model 2. However,
since $k_{1}(t)$ is not a function of $\gamma$, as illustrated in
\eqref{eq:k1_model_3}, $k_{1}(t)$ does not change with $\gamma$. }
\item \textcolor{black}{Now consider the case where $T$ increase while
$\gamma$ is fixed. For $\gamma=1$, $k_{1}(T)=0.5$ and $k_{2}(T)=0.64$
for both $T=1$ and $T=3$. For $\gamma=3$, $k_{1}(T)=0.5$ and $k_{2}(T)=0.2133$
for both $T=1$ and $T=3$. As can be seen in Figure \ref{fig: k_1_k_2_model_3_2BM_demonstration},
both $k_{1}(t)$ and $k_{2}(t)$ decreases as $T$ increases from
1 to 3. When approaching the expiry date, since the investment time
horizon is becoming shorter, an investor will put more and more money
on the stock with higher appreciation rate in order to have a higher
expected wealth on the expiry date. }
\item \textcolor{black}{Suppose that for all $t\in[0,T]$, $X_{t}\equiv1$,
i.e., the initial wealth at any time $t$ is 1. The amount of money
invested in stock 1 at time $t$ is thus $\hat{u}_{1}(t,x)=k_{1}(t)+k_{2}(t)$.
It is observed that $\hat{u}_{1}(t,1)$ is always larger than 0.5.
Therefore, we invest more than a half of our wealth into stock 1.
Furthermore, the total amount of money invested in stock 1 is larger
than 1 for all $t\in[0,T]$ when $T=1$ and $\gamma=1$. For $T=3$
and $\gamma=1$, this also happens when approaching the expiry date.
In this case, one holds a long position of stock 1 and a short position
of stock 2. However, as the investor becomes more and more risk averse,
the amount of money invested in stock 1 reduces and is less than 1
when $\gamma=3$.  In this circumstance, an investor holds long positions
of both stock 1 and stock 2. }
\end{itemize}
\end{itemize}

\subsection{A numerical algorithm for $k_{1}$ and $k_{2}$}

The algorithm is an analog of Björk et al. (2014) in which a 1-dim
ODE is dealt with.
\begin{thm}
\label{thm:Convergence_Speed_k1_k2}Suppose the sequence $\{k_{1}^{(n)}\}$
is constructed by
\begin{align}
k_{1}^{(0)}(t) & =1,\label{eq:Iteration_Scheme_k10}\\
k_{1}^{(n)}(t) & =\frac{1}{\left(\sigma_{11}-\sigma_{21}\right)^{2}+\left(\sigma_{12}-\sigma_{22}\right)^{2}}\Bigg\{-\sigma_{21}\left(\sigma_{11}-\sigma_{21}\right)-\sigma_{22}\left(\sigma_{12}-\sigma_{22}\right)\label{eq:Iteration_Scheme_k1n}\\
 & \;\;\;\;+\left(\alpha_{1}-\alpha_{2}\right)\left(e^{-\intop_{t}^{T}\left[\left\{ \sigma_{21}+k_{1}^{(n-1)}(s)(\sigma_{11}-\sigma_{21})\right\} ^{2}+\left\{ \sigma_{22}+k_{1}^{(n-1)}(s)(\sigma_{12}-\sigma_{22})\right\} ^{2}\right]ds}-1\right)\Bigg\},\nonumber
\end{align}
 for $n=1,\ldots$. Then we have
\[
\big|k_{1}^{(n)}(t)-k_{1}(t)\big|\le\sum_{i=n}^{\infty}\frac{1}{i!}K^{i+1}(T-t)^{i},\quad n=1,2,\ldots.
\]
With the known $k_{1}(t)$, we can construct another sequence $\{k_{2}^{(n)}\}$:
\begin{align}
k_{2}^{(0)}(t) & =1,\label{eq:Iteration_Scheme_k20}\\
k_{2}^{(n)}(t) & =\frac{\alpha_{1}-\alpha_{2}}{\left(\sigma_{11}-\sigma_{21}\right)^{2}+\left(\sigma_{12}-\sigma_{22}\right)^{2}}\bigg\{\frac{1}{\gamma}I_{1}(t,T)I_{2}(t,T)+\int_{t}^{T}I_{1}(t,v)I_{3}(t,v)k_{2}^{(n-1)}(v)dv\bigg\}.\label{eq:Iteration_Scheme_k2n}
\end{align}
For this sequence, we have
\[
\big|k_{2}^{(n)}(t)-k_{2}(t)\big|\le\sum_{j=n}^{\infty}\frac{1}{j!}K^{j+1}(T-t)^{j},\quad n=1,2,\ldots.
\]
\end{thm}
\begin{proof}
See Appendix \ref{sec: Appendix_C}
\end{proof}

\section{Numerical results\label{sec:Numerical_Results}}

\subsection{The solution of model 3}

The parameters chosen are $\alpha_{1}=0.2$, $\sigma_{11}=0.3$, $\sigma_{12}=0$,
$\alpha_{2}=0.12$, $\sigma_{21}=0$, $\sigma_{22}=0.2$. Figure \ref{fig: k1_k2_Model3_various_gamma}
and Figure \ref{fig: k1_k2_Model3_various_T} are the plots of $k_{1}(t)$
and $k_{2}(t)$ with various $\gamma=1,3,5,10$ and $T=0.5,1,5,10$.
The dynamics of stock 1 and stock 2 are
\[
\begin{cases}
dS_{1t}=S_{1t}(0.20dt+0.3dW_{1t}),\\
dS_{2t}=S_{2t}(0.12dt+0.2dW_{2t}).
\end{cases}
\]

From \textbf{Figure \ref{fig: k1_k2_Model3_various_gamma}}, we can
see that $k_{2}(t)$ decreases as the risk aversion coefficient $\gamma$
increases. This is also reflected in the structures of the optimal
allocations and reward functions for model 2 and model 3. Frankly
speaking, model 3 converges to model 2 as $\gamma$ converges to infinity.
\begin{figure}[h]

\begin{centering}
\includegraphics[scale=0.4]{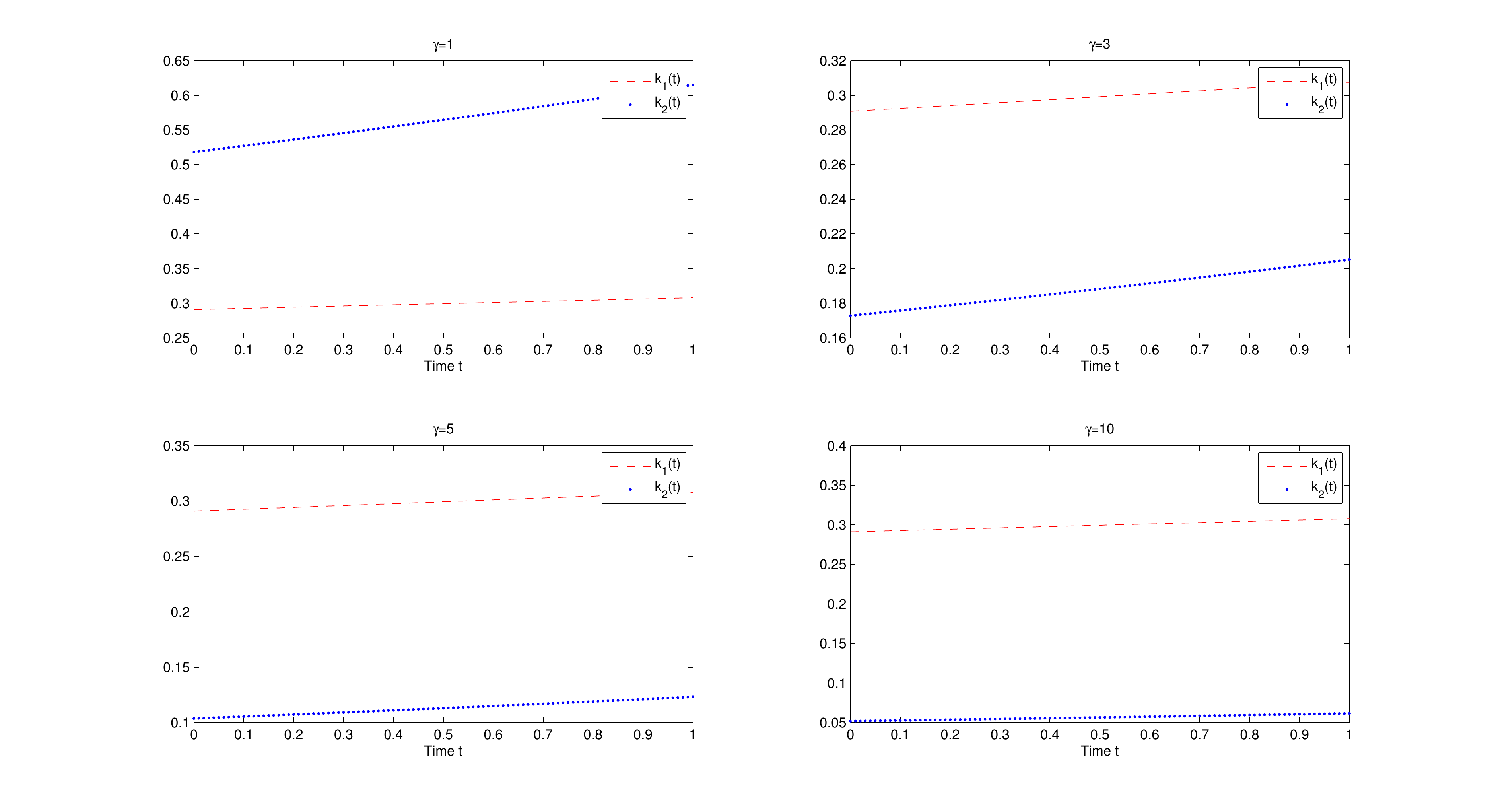}
\par\end{centering}

\caption{The functions $k_{1}(t)$ and $k_{2}(t)$ for various choices of $\gamma$
with $T=1$.\label{fig: k1_k2_Model3_various_gamma}}
\end{figure}

On the other hand, \textbf{Figure \ref{fig: k1_k2_Model3_various_T}}
shows that both $k_{1}(t)$ and $k_{2}(t)$ decrease as the terminal
time $T$ increases. However, for different time horizons, both $k_{1}(t)$
and $k_{2}(t)$ coincide at the date of maturity.
\begin{figure}[h]
\begin{centering}
\includegraphics[scale=0.4]{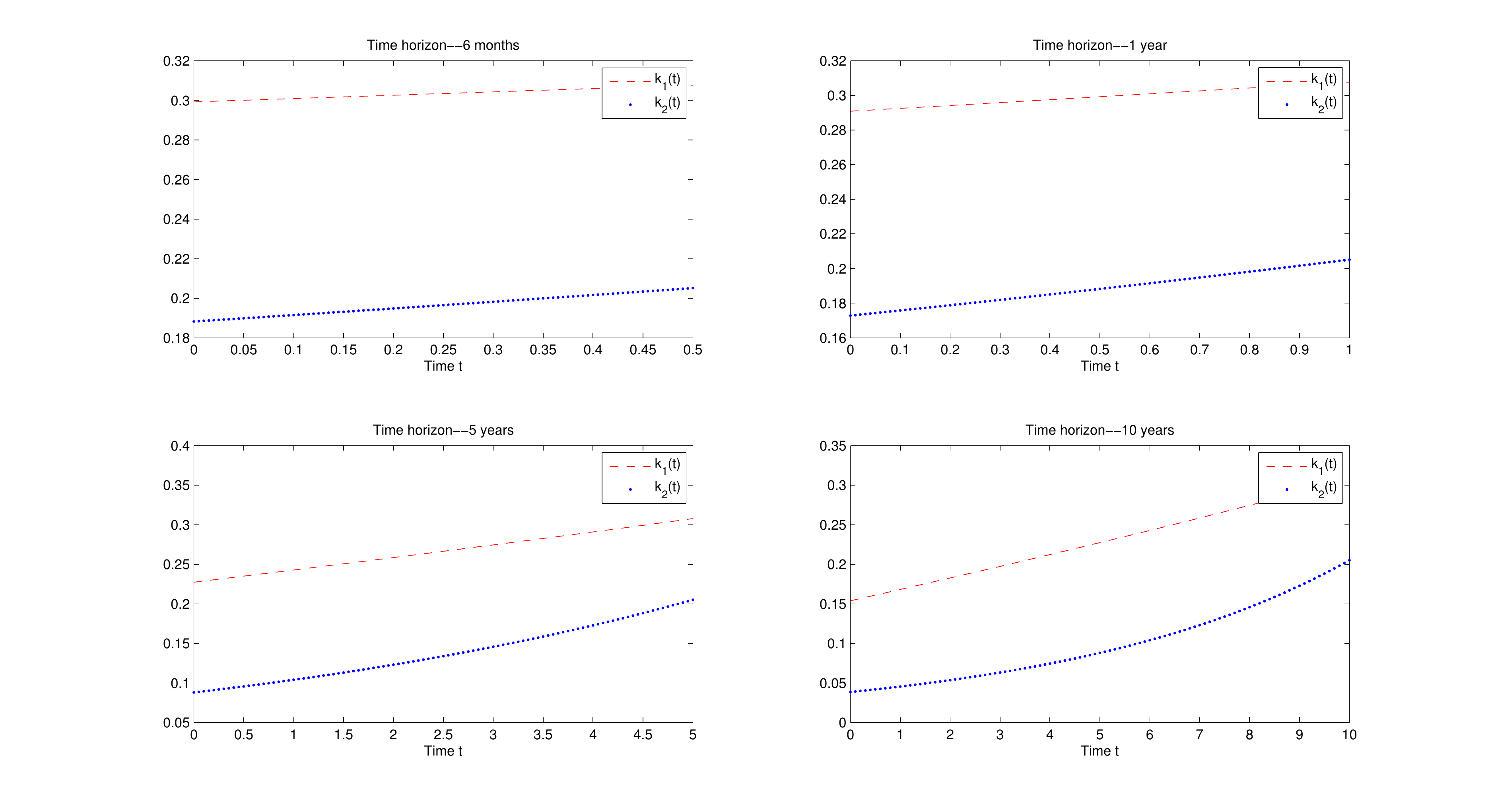}
\par\end{centering}

\caption{The functions $k_{1}(t)$ and $k_{2}(t)$ for various choices of $T$
with $\gamma=3$.\label{fig: k1_k2_Model3_various_T}}
\end{figure}

\textbf{Figure \ref{fig:u1_u2_Model3_3d_plot}} shows the effect of
$t$ and $x$ on $\hat{u}_{1}(t,x)$ and $\hat{u}_{2}(t,x)=x-\hat{u}_{1}(t,x)$,
which are the amounts of money invested in stock 1 and stock 2, respectively.
The terminal time $T=10$ and the risk aversion coefficient $\gamma=3$.
From Figure \ref{sec:Problem_Formulation}, we can see that $0<k_{1}(t)<1$
for all $t\in[0,10]$ and thus both $\hat{u}_{1}(t,x)$ and $\hat{u}_{2}(t,x)$
increase as the value of $x$ increase. Moreover, since both $k_{1}(t)$
and $k_{2}(t)$ are increasing with time $t$, therefore for each
fixed wealth state $x$, $\hat{u}_{1}(t,x)$ is an increasing function
and $\hat{u}_{2}(t,x)$ is a decreasing function of $t$.
\begin{figure}[h]
\begin{centering}
\includegraphics[scale=0.4]{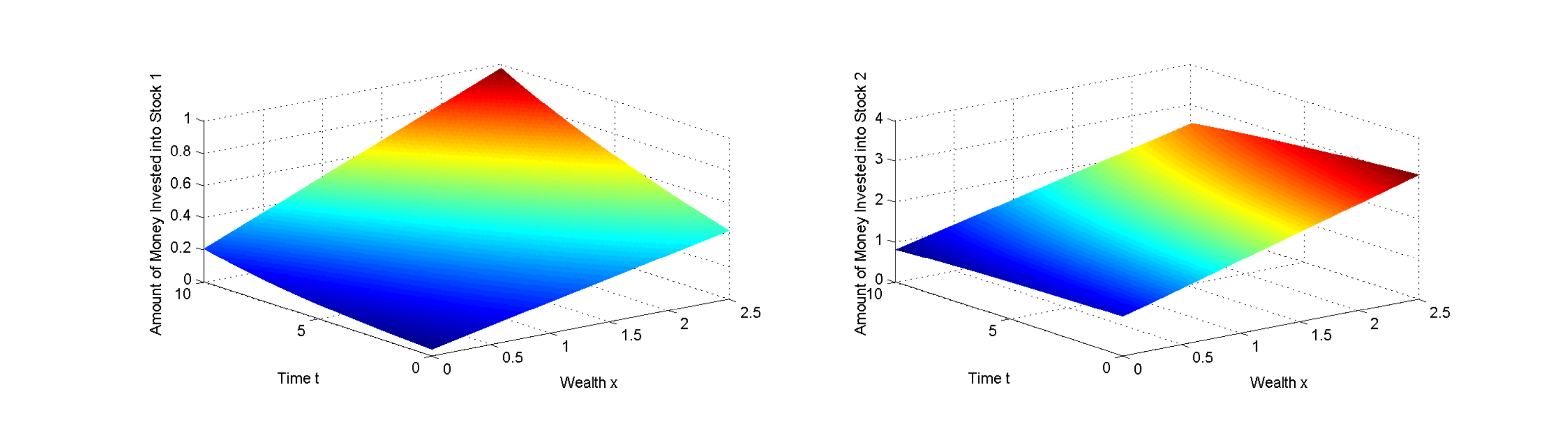}
\par\end{centering}

\caption{The effect of parameters on the amount of money invested in stock
1 and stock 2 in model 3.\label{fig:u1_u2_Model3_3d_plot}}
\end{figure}

\subsection{Comparison of three models}

Here, the parameters are $\gamma=3$, $T=1$, $\alpha_{1}=0.2$, $\sigma_{11}=0.3$,
$\sigma_{12}=0$, $\alpha_{2}=0.12$, $\sigma_{21}=0$, $\sigma_{22}=0.2$
and $r=0.04$. In order to compare the investment strategies of the
three models at different wealth level, we fix $t=0$ and plot Figure
\ref{fig:Amount_Proportion_Various_x_3Models_Comparison} and Figure
\ref{fig: Mean_Variance_various_x_3Models_Comparison}. From \textbf{Figure
\ref{fig:Amount_Proportion_Various_x_3Models_Comparison}}, we can
see that in model 2, since the objective in this model is to minimize
the variance of the expected wealth at the terminal time $T$, the
amount of money invested in stock 2 is larger than that invested in
stock 1 for all wealth level since stock 2 has the variance which
is less than that of stock 1. In model 3, the amount of money invested
in stock 1 is larger than that in model 2 while the one invested in
stock 2 is less than that in model 3. This is because besides minimizing
the variance, maximizing the expected wealth is also our objective
in model 3. Therefore, one will invest more money into the stock with
larger price of market risk.\footnote{One can calculate that the price of market price for stock 1 and stock
2 are $(\alpha_{1}-r)/\sigma_{11}=0.5333$ and $(\alpha_{2}-r)/\sigma_{22}=0.4000$.} In addition, under model 3, we can see that when with low initial
wealth, one even holds short position of stock 2. Furthermore, from\textbf{
Figure \ref{fig: Mean_Variance_various_x_3Models_Comparison}}, the
expected wealth and the conditional variance of wealth of model 3
are both larger than the ones of model 2.
\begin{figure}[h]
\begin{centering}
\includegraphics[clip,scale=0.4]{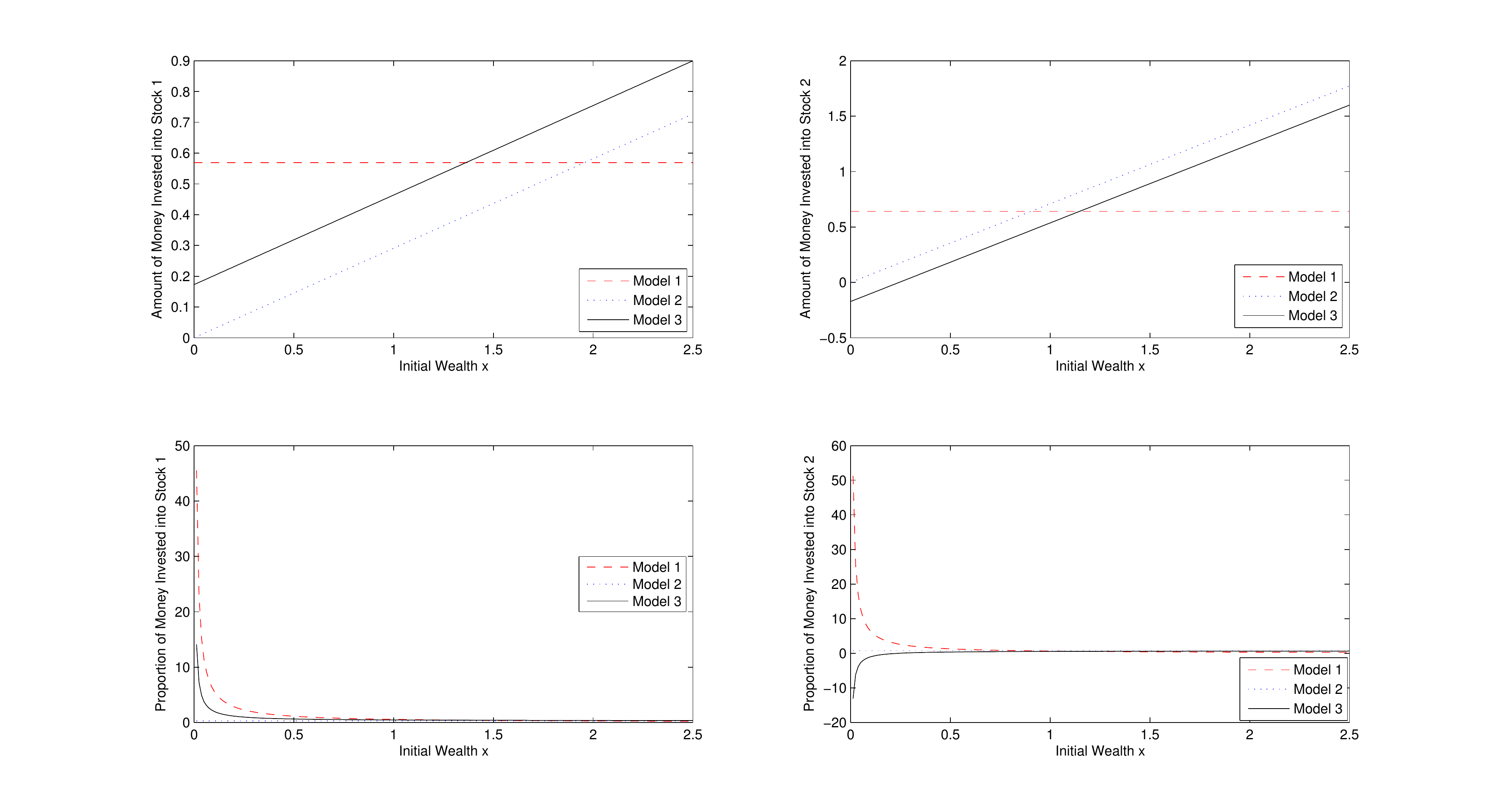}
\par\end{centering}

\caption{The amounts and proportions of money invested in stock 1 and 2 in
the three models with $t=0$.\label{fig:Amount_Proportion_Various_x_3Models_Comparison}}
\end{figure}
\begin{figure}[h]
\begin{centering}
\includegraphics[scale=0.4]{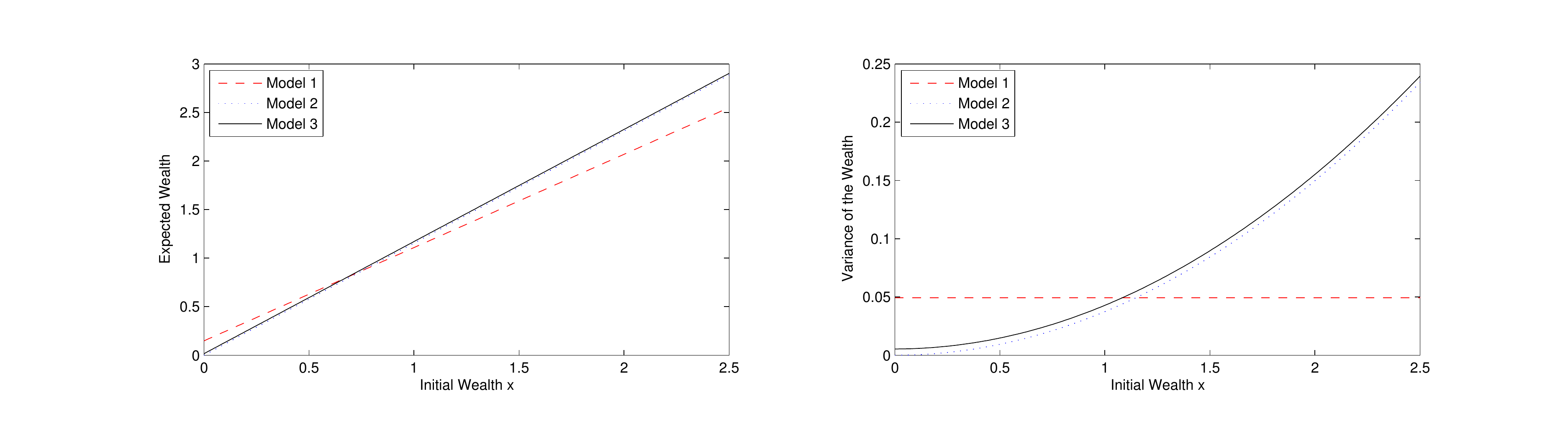}
\par\end{centering}

\centering{}\caption{The expected wealths and the conditional variances of the wealth in
the three models with $t=0$.\label{fig: Mean_Variance_various_x_3Models_Comparison}}
\end{figure}

For demonstration, we simulate two paths of the price, one for stock
1 and the other for stock 2. We then calculate the amounts and proportions
of money invested into these two stocks, the wealth processes, the
expected wealths and the variances for the three models based on these
two paths. The amounts and proportions of money invested in stock
1 and 2 can be found in \textbf{Figure \ref{fig:Simulated_Amount_of_Money_3Models}}.
\begin{figure}[h]
\begin{centering}
\includegraphics[scale=0.4]{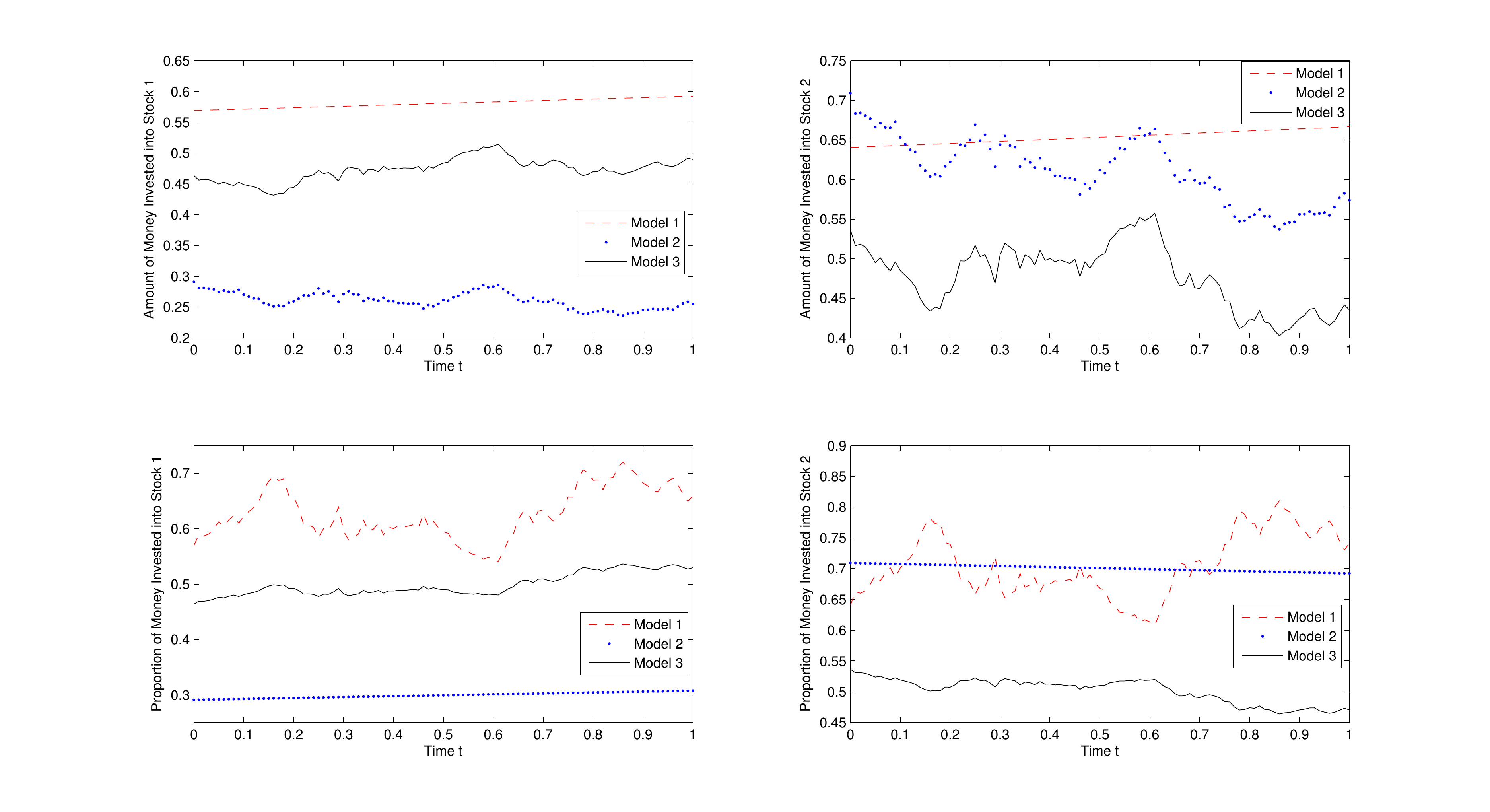}
\par\end{centering}

\caption{Simulated paths of the amounts and proportions of money invested into
stock 1 and 2 for the three models. \label{fig:Simulated_Amount_of_Money_3Models}}
\end{figure}

Observing that stock 1 has a larger volatility than that of stock
2, it is not surprising that the amount of money invested in stock
1 is the least for model 2 because it puts all efforts to minimize
the variance. And obviously in model 1, the investor borrows money
to make risky investment since the sum of money invested into stocks
exceeds 1.

\textbf{Figure \ref{fig:Simulated_Mean_Variance_3Models}} is the
simulated paths of the wealth processes, means, variances and objective
functions for the three models. Model 3 has the largest expected wealth
with middle variance. Model 2 has the lowest variance with also lowest
expected wealth. Overall, model 3 performs out of the three models
with the largest reward in a bull market.\footnote{Here we have assumed both $\alpha_{1}$ and $\alpha_{2}$ are greater
than the riskless interest rate $r$.}
\begin{figure}[h]
\begin{centering}
\includegraphics[scale=0.4]{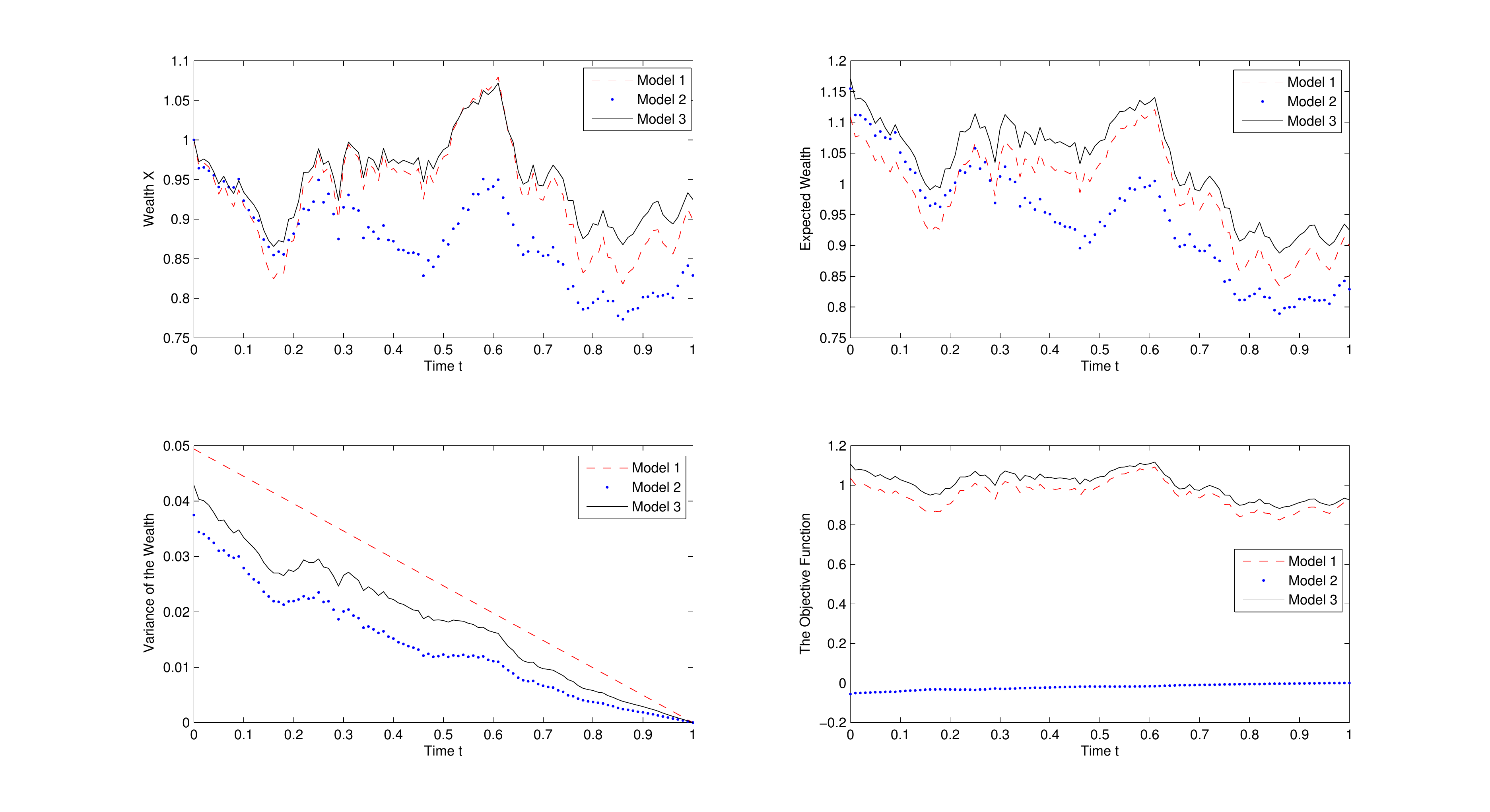}
\par\end{centering}

\caption{Simulated paths of the expected wealths, the conditional variances
of wealths and the reward functions for the three models. \label{fig:Simulated_Mean_Variance_3Models}}
\end{figure}

\section{Conclusion\label{sec:Conclusion}}

In this paper, we construct the optimal time-consistent portfolio
selection strategy for correlated risky assets explicitly without
risk-free asset under the game theoretic framework. The key idea is
an application of the extended HJB system developed in Björk and Murgoci
(2010). The equilibrium control is linear in wealth. If a risk-free
asset is involved, the equilibrium control has zero slope (i.e., it
is independent of the current wealth) and is consistent with the ones
in Basak and Chabakauri (2010) and Björk and Murgoci (2010). Therefore,
model 3 in Section \ref{sec:3_Models_Description} can be considered
as an extension of their models. On the other hand, as risk aversion
approaches infinity, the intercept term of the optimal control tends
to 0. In this sense, the equilibrium control in model 3 converges
to that in model 2. Theorem \ref{thm:ODE_Unique_Solution} provides
the existence and uniqueness of the optimal solution. We also present
an iterative scheme for the determination of the optimal solution
and its convergence speed is given in Theorem \ref{thm:Convergence_Speed_k1_k2}.
We conduct numerical studies for the comparisons of the amounts and
proportions of money invested in the assets, the expected values of
the terminal wealth, the conditional variances and the objectives
functions of the three models. Comparisons demonstrate that model
3 performs better than the previous two in a bull market.

All the three models choose optimal strategies according to the prices
of market risk, i.e., assets with higher prices of market risk are
allocated more and assets with the same level price of market risk
share the same allocation.

Different from Björk et al. (2014), the optimal allocations in model
2 and model 3 depend on the initial wealth although the risk aversion
coefficient $\gamma$ is a constant. The optimal reward functions
are quadratic in initial wealth $x$. Usually a risk-seeking investor
would put all his money into risky assets when he is optimistic in
the market.

\clearpage{}
\begin{appendices}

\section{}\label{sec: Appendix_A}

\setcounter{equation}{0} \renewcommand{\theequation}{A.\arabic{equation}}

\textbf{Proof of Theorem \ref{thm:optimal_allocation_u_model_3}.
}The dynamic of the value process of the portfolio is given by
\begin{align*}
dX_{t} & =[\{\alpha_{2}+k_{1}(t)(\alpha_{1}-\alpha_{2})\}X_{t}+k_{2}(t)(\alpha_{1}-\alpha_{2})]dt\\
 & \;\;\;\;+\big[\{\sigma_{21}+k_{1}(t)(\sigma_{11}-\sigma_{21})\}X_{t}+k_{2}(t)(\sigma_{11}-\sigma_{21})\}\big]dW_{1t}\\
 & \;\;\;\;+\big[\{\sigma_{22}+k_{1}(t)(\sigma_{12}-\sigma_{22})\}X_{t}+k_{2}(t)(\sigma_{12}-\sigma_{22})\}\big]dW_{2t},\\
dX_{t}^{2} & =[2\{\alpha_{2}+k_{1}(t)(\alpha_{1}-\alpha_{2})\}+\{\sigma_{21}+k_{1}(t)(\sigma_{11}-\sigma_{21})\}^{2}+\{\sigma_{22}+k_{1}(t)(\sigma_{12}-\sigma_{22})\}^{2}]X_{t}^{2}dt\\
 & \;\;\;\;+2[(\alpha_{1}-\alpha_{2})+\{\sigma_{21}+k_{1}(t)(\sigma_{11}-\sigma_{21})\}(\sigma_{11}-\sigma_{21})+\{\sigma_{22}+k_{1}(t)(\sigma_{12}-\sigma_{22})\}(\sigma_{12}-\sigma_{22})]k_{2}(t)X_{t}dt\\
 & \;\;\;\;+\{(\sigma_{11}-\sigma_{21})^{2}+(\sigma_{12}-\sigma_{22})^{2}\}k_{2}^{2}(t)dt\\
 & \;\;\;\;+2[\{\sigma_{21}+k_{1}(t)(\sigma_{11}-\sigma_{21})\}X_{t}+k_{2}(t)(\sigma_{11}-\sigma_{21})]X_{t}dW_{1t}\\
 & \;\;\;\;+2\big[\{\sigma_{22}+k_{1}(t)(\sigma_{12}-\sigma_{22})\}X_{t}+k_{2}(t)(\sigma_{12}-\sigma_{22})\big]X_{t}dW_{2t},
\end{align*}
Denote $\mu_{t,x}(T)=\mathrm{E}_{t,x}(X_{T}^{\hat{u}})$ and $q_{t,x}(T)=\mathrm{E}_{t,x}\{(X_{T}^{\hat{u}})^{2}\}$.
By taking expectations on both sides of these two equations, we have

\begin{align}
\mu_{t,x}(T) & =x+\int_{t}^{T}[\{\alpha_{2}+k_{1}(s)(\alpha_{1}-\alpha_{2})\}\mu_{t,x}(s)+k_{2}(s)(\alpha_{1}-\alpha_{2})]ds,\label{eq:Integral_Form_First_Moment}\\
q_{t,x}(T) & =x^{2}+\int_{t}^{T}[2\{\alpha_{2}+k_{1}(s)(\alpha_{1}-\alpha_{2})\}+\{\sigma_{21}+k_{1}(s)(\sigma_{11}-\sigma_{21})\}^{2}\nonumber \\
 & \;\;\;\;\;\;\;\;+\{\sigma_{22}+k_{1}(s)(\sigma_{12}-\sigma_{22})\}^{2}]q_{t,x}(s)ds\nonumber \\
 & \;\;\;\;+\int_{t}^{T}2[(\alpha_{1}-\alpha_{2})+\{\sigma_{21}+k_{1}(s)(\sigma_{11}-\sigma_{21})\}(\sigma_{11}-\sigma_{21})\nonumber \\
 & \;\;\;\;\;\;\;\;+\{\sigma_{22}+k_{1}(s)(\sigma_{12}-\sigma_{22})\}(\sigma_{12}-\sigma_{22})]k_{2}(s)\mu_{t,x}(s)ds\nonumber \\
 & \;\;\;\;+\int_{t}^{T}\{(\sigma_{11}-\sigma_{21})^{2}+(\sigma_{12}-\sigma_{22})^{2}\}k_{2}^{2}(s)ds.\label{eq:Integral_Form_Second_Moment}
\end{align}

Equation \eqref{eq:Integral_Form_First_Moment} is a general linear
differential equation. The solution of this equation is

\[
\mu_{t,x}(T)=\{I_{1}(t,T)\}^{-1}\bigg\{ x+(\alpha_{1}-\alpha_{2})\int_{t}^{T}I_{1}(t,v)k_{2}(t,v)dv\bigg\}.
\]

With this known $\mu_{t,x}(T)$, equation \eqref{eq:Integral_Form_Second_Moment}
is also a general linear differential equation and its solution is
\begin{align*}
q_{t,x}(T) & =x^{2}I_{1}^{-2}(t,T)I_{2}^{-1}(t,T)\\
 & \;\;\;\;+2xI_{1}^{-2}(t,T)I_{2}^{-1}(t,T)\int_{t}^{T}[(\alpha_{1}-\alpha_{2})+\{\sigma_{21}+(\sigma_{11}-\sigma_{21})k_{1}(v)\}(\sigma_{11}-\sigma_{21})\\
 & \;\;\;\;\;\;\;\;+\{\sigma_{22}+(\sigma_{12}-\sigma_{22})k_{1}(v)\}(\sigma_{12}-\sigma_{22})]I_{1}(t,v)I_{2}(t,v)k_{2}(v)dv\\
 & \;\;\;\;+2(\alpha_{1}-\alpha_{2})I_{1}^{-2}(t,T)I_{2}^{-1}(t,T)\int_{t}^{T}I_{1}(t,v)I_{2}(t,v)k_{2}(v)[(\alpha_{1}-\alpha_{2})\\
 & \;\;\;\;\;\;\;\;+\{\sigma_{21}+k_{1}(v)(\sigma_{11}-\sigma_{21})\}(\sigma_{11}-\sigma_{21})\\
 & \;\;\;\;\;\;\;\;+\{\sigma_{22}+(\sigma_{12}-\sigma_{22})k_{1}(v)\}(\sigma_{12}-\sigma_{22})]\int_{t}^{v}k_{2}(w)I_{1}(t,w)dwdv\\
 & \;\;\;\;+I_{1}^{-2}(t,T)I_{2}^{-1}(t,T)\{(\sigma_{11}-\sigma_{21})^{2}+(\sigma_{12}-\sigma_{22})^{2}\}\int_{t}^{T}I_{1}^{2}(t,v)I_{2}(t,v)k_{2}^{2}(v)dv.
\end{align*}

Therefore, the expectation and variance are
\begin{align*}
\mathrm{E}_{t,x}(X_{T}^{\hat{u}}) & =\mu_{t,x}(T)\\
 & =\{I_{1}(t,T)\}^{-1}\bigg\{ x+(\alpha_{1}-\alpha_{2})\int_{t}^{T}I_{1}(t,v)k_{2}(v)dv\bigg\},\\
\mathrm{Var}_{t,x}(X_{T}^{\hat{u}}) & =q_{t,x}(T)-\{\mu_{t,x}(T)\}^{2}\\
 & =\{I_{1}(t,T)\}^{-2}\big\{ c_{0}(t)x^{2}+c_{1}(t)x+c_{2}(t)\big\},
\end{align*}

where
\begin{align*}
c_{0}(t) & =\{I_{2}(t,T)\}^{-1}-1,\\
c_{1}(t) & =-2I_{2}^{-1}(t,T)\int_{t}^{T}I_{1}(t,v)I_{3}(t,v)k_{2}(v)dv,\\
c_{2}(t) & =\{I_{2}(t,T)\}^{-1}\bigg[\{(\sigma_{11}-\sigma_{21})^{2}+(\sigma_{12}-\sigma_{22})^{2}\}\int_{t}^{T}I_{1}^{2}(t,v)I_{2}(t,v)k_{2}^{2}(v)dv\\
 & \;\;\;\;-2(\alpha_{1}-\alpha_{2})\int_{t}^{T}I_{1}(t,v)I_{3}(t,v)k_{2}(v)\bigg\{\int_{t}^{v}I_{1}(t,w)k_{2}(w)dw\bigg\} dv\bigg].
\end{align*}

The functions $g(t,x)$ and $V(t,x)$ are given by
\begin{align}
g(t,x) & =\mathrm{E}_{t,x}(X_{T}^{\hat{u}})\nonumber \\
 & =\{I_{1}(t,T)\}^{-1}x+(\alpha_{1}-\alpha_{2})\{I_{1}(t,T)\}^{-1}\int_{t}^{T}I_{1}(t,v)k_{2}(v)dv,\label{eq:Exp_Model_3}\\
V(t,x) & =\mathrm{E}_{t,x}(X_{T}^{\hat{u}})-\frac{\gamma}{2}\mathrm{Var}_{t,x}(X_{T}^{\hat{u}})\nonumber \\
 & =-\frac{\gamma}{2}\{I_{1}(t,T)\}^{-2}c_{0}(t)x^{2}+\big[\{I_{1}(t,T)\}^{-1}-\frac{\gamma}{2}\{I_{1}(t,T)\}^{-2}c_{1}(t)\big]x\nonumber \\
 & \;\;\;\;+(\alpha_{1}-\alpha_{2})\{I_{1}(t,T)\}^{-1}\int_{t}^{T}I_{1}(t,v)k_{2}(v)dv-\frac{\gamma}{2}\{I_{1}(t,T)\}^{-2}c_{2}(t).\label{eq:Value_Function_Model_3}
\end{align}
Here, $g(t,\cdot)$ is a linear function while $V(t,\cdot)$ is a
linear--quadratic function of the current state. By the first order
condition and the HJB system \eqref{eq:HJB_Model_3}, the optimal
allocation $\hat{u}(t,x)$ has the expression:

\[
\hat{u}(t,x)=\frac{(\alpha_{2}-\alpha_{1})V_{x}-\{\sigma_{21}(\sigma_{11}-\sigma_{21})+\sigma_{22}(\sigma_{12}-\sigma_{22})\}(V_{xx}-\gamma g_{x}^{2})x}{\{(\sigma_{11}-\sigma_{21})^{2}+(\sigma_{12}-\sigma_{22})^{2}\}(V_{xx}-\gamma g_{x}^{2})}.
\]
By substituting \eqref{eq:Exp_Model_3} and \eqref{eq:Value_Function_Model_3}
into this equation and after some tedious algebra, we obtain the ODE
system \eqref{eq:k1_model_3}--\eqref{eq:k2_model_3} for $k_{1}(t)$
and $k_{2}(t)$.

\section{}\label{sec: Appendix_B}

\setcounter{equation}{0} \renewcommand{\theequation}{B.\arabic{equation}}

\textbf{Proof of Position \ref{thm:ODE_Unique_Solution}.} We assume
that $\alpha_{1}>\alpha_{2}$ and the situation where $\alpha_{1}<\alpha_{2}$
can be similarly treated. We first prove that the integral equation
(\ref{eq:k1_model_3}) admits a unique solution $k_{1}\in C[0,T]$.
Construct a sequence

\begin{align}
k_{1}^{(0)}(t) & =1,\nonumber \\
k_{1}^{(n)}(t) & =\frac{1}{\left(\sigma_{11}-\sigma_{21}\right)^{2}+\left(\sigma_{12}-\sigma_{22}\right)^{2}}\Bigg\{-\sigma_{21}\left(\sigma_{11}-\sigma_{21}\right)-\sigma_{22}\left(\sigma_{12}-\sigma_{22}\right)\label{eq:sequence_k1_n}\\
 & \;\;\;\;+\left(\alpha_{1}-\alpha_{2}\right)\bigg(e^{-\intop_{t}^{T}\left[\left\{ \sigma_{21}+k_{1}^{(n-1)}(s)(\sigma_{11}-\sigma_{21})\right\} ^{2}+\left\{ \sigma_{22}+k_{1}^{(n-1)}(s)(\sigma_{12}-\sigma_{22})\right\} ^{2}\right]ds}-1\bigg)\Bigg\},\nonumber
\end{align}
for $n=1,\ldots$.

For all $t\in[0,T]$, from
\[
0\le e^{-\intop_{t}^{T}\left[\left\{ \sigma_{21}+k_{1}^{(n-1)}(s)(\sigma_{11}-\sigma_{21})\right\} ^{2}+\left\{ \sigma_{22}+k_{1}^{(n-1)}(s)(\sigma_{12}-\sigma_{22})\right\} ^{2}\right]ds}\le1,
\]
we have

\begin{align*}
k_{1}^{(n)}(t) & \ge-\frac{1}{\left(\sigma_{11}-\sigma_{21}\right)^{2}+\left(\sigma_{12}-\sigma_{22}\right)^{2}}\Bigg\{\sigma_{21}\left(\sigma_{11}-\sigma_{21}\right)+\sigma_{22}\left(\sigma_{12}-\sigma_{22}\right)+(\alpha_{1}-\alpha_{2})\Bigg\},\\
k_{1}^{(n)}(t) & \le-\frac{1}{\left(\sigma_{11}-\sigma_{21}\right)^{2}+\left(\sigma_{12}-\sigma_{22}\right)^{2}}\Bigg\{\sigma_{21}\left(\sigma_{11}-\sigma_{21}\right)+\sigma_{22}\left(\sigma_{12}-\sigma_{22}\right)\Bigg\}.
\end{align*}
Therefore, the sequence $\{k_{1}^{(n)}\}$ is uniformly bounded in
$C[0,T]$.

We now consider the sequence $\big\{\dot{k}_{1}^{(n)}\big\}$ where
$\dot{k}_{1}^{(n)}=dk_{1}^{(n)}(t)/dt$. The derivative $\dot{k}_{1}^{(n)}$
has the expression:
\begin{align}
\dot{k}_{1}^{(n)}(t) & =\frac{\alpha_{1}-\alpha_{2}}{\left(\sigma_{11}-\sigma_{21}\right)^{2}+\left(\sigma_{12}-\sigma_{22}\right)^{2}}\bigg[\bigg\{\sigma_{21}+k_{1}^{(n-1)}(t)(\sigma_{11}-\sigma_{21})^{2}\bigg\}+\bigg\{\sigma_{22}+\label{eq:diff_k1}\\
 & \;\;\;\;k_{1}^{(n-1)}(t)(\sigma_{12}-\sigma_{22})^{2}\bigg\}\bigg]e^{-\intop_{t}^{T}\left[\left\{ \sigma_{21}+k_{1}^{(n-1)}(s)(\sigma_{11}-\sigma_{21})\right\} ^{2}+\left\{ \sigma_{22}+k_{1}^{(n-1)}(s)(\sigma_{12}-\sigma_{22})\right\} ^{2}\right]ds},\nonumber
\end{align}
Since we have proved that $\{k_{1}^{(n)}\}$ is uniformly bounded
in $C[0,T]$, from (\ref{eq:diff_k1}), we can conclude that the sequence
$\big\{\dot{k}_{1}^{(n)}\big\}$ is also uniformly bounded in $C[0,T]$.
Denote $|k_{1}^{(n)}|<M_{1}$ for all $n$ and all $t\in[0,T]$. Therefore,
for any $t_{1},t_{2}\in[0,T]$ and $t_{1}<t_{2}$, we have
\begin{align*}
\big|k_{1}^{(n)}(t_{1})-k_{1}^{(n)}(t_{2})\big| & =\bigg|\int_{0}^{1}\frac{d}{ds}k_{1}^{(n)}(t_{1}+s(t_{2}-t_{1}))ds\bigg|\\
 & =\bigg|(t_{2}-t_{1})\int_{0}^{1}\dot{k}_{1}^{(n)}(t_{1}+s(t_{2}-t_{1}))ds\bigg|\\
 & \le(t_{2}-t_{1})\max_{0\le t\le T}\big|\dot{k}_{1}^{(n)}(t)\big|=M_{1}(t_{2}-t_{1}).
\end{align*}
Therefore, the sequence $\{k_{1}^{(n)}\}$ is also equicontinuous.
According to Arzela--Ascoli Theorem, there exists a subsequence of
$\{k_{1}^{(n)}\}$, $\{k_{1}^{(n_{i})}\}$, and a $k_{1}\in C[0,T]$
such that $k_{1}^{(n_{i})}\to k_{1}$ as $i\to\infty$. Since $\{k_{1}^{(n_{i})}\}$
satisfies (\ref{eq:sequence_k1_n}), by letting $i\to\infty$, we
can conclude that $k_{1}$ is a solution to (\ref{eq:sequence_k1_n}).

For the uniqueness of the solution, suppose $k_{1}$ and $k_{2}$
are two solutions to equation (\ref{eq:sequence_k1_n}). Since $k_{1}$
and $l_{1}$ are bounded in $[0,T]$, therefore the functions
\[
-\intop_{t}^{T}\left[\left\{ \sigma_{21}+k_{1}(s)(\sigma_{11}-\sigma_{21})\right\} ^{2}+\left\{ \sigma_{22}+k_{1}(s)(\sigma_{12}-\sigma_{22})\right\} ^{2}\right]ds
\]
 and

\[
-\intop_{t}^{T}\left[\left\{ \sigma_{21}+l_{1}(s)(\sigma_{11}-\sigma_{21})\right\} ^{2}+\left\{ \sigma_{22}+l_{1}(s)(\sigma_{12}-\sigma_{22})\right\} ^{2}\right]ds
\]
are also bounded for all $t\in[0,T]$. Since the function $f(x)=e^{x}$
is Lipschitz on bounded set, it is easy to derive that
\begin{align*}
|k_{1}(t)-l_{1}(t)| & \le M_{2}\int_{t}^{T}|k_{1}(s)-l_{1}(s)|ds.
\end{align*}
This Gronwal inequality implies that $k_{1}(t)=l_{1}(t)$ for all
$t\in[0,T]$.

We thus proved that equation (\ref{eq:k1_model_3}) admits a unique
solution $k_{1}\in C[0,T]$. For equation (\ref{eq:k2_model_3}),
denote $\lambda=(\alpha_{1}-\alpha_{2})/\{\left(\sigma_{11}-\sigma_{21}\right)^{2}+\left(\sigma_{12}-\sigma_{22}\right)^{2}\}$,
$\phi(t)=\gamma^{-1}\lambda I_{1}(t,T)I_{2}(t,T)$ and $L(t,v)=I_{1}(t,v)I_{3}(t,v)$.
Then we have
\begin{equation}
k_{2}(t)=\phi(t)+\lambda\int_{t}^{T}L(t,v)k_{2}(v)dv.\label{eq:Volterra_equation}
\end{equation}
Equation (\ref{eq:Volterra_equation}) is a Volterra integral equation
of the second kind. Consider the mapping $F:C[0,T]\to C[0,T]$,
\[
Fk_{2}(t)=\phi(t)+\lambda\int_{t}^{T}L(t,v)k_{2}(v)dv.
\]
Then for all $k_{2},l_{2}\in C[0,T]$,
\begin{align*}
|Fk_{2}(t)-Fl_{2}(t)| & =|\lambda|\bigg|\int_{t}^{T}L(t,v)\{k_{2}(v)-l_{2}(v)\}dv\bigg|\\
 & \le|\lambda|M_{3}(T-t)\max_{t\le s\le T}|k_{2}(s)-l_{2}(s)|
\end{align*}
where $M_{3}=\max_{0\le t_{1},t_{2}\le T}|L(t_{1},t_{2})|$.

Therefore,
\begin{align*}
|F^{2}k_{2}(t)-F^{2}l_{2}(t)| & =|\lambda|\bigg|\int_{t}^{T}L(t,v)\{Fk_{2}(v)-Fl_{2}(v)\}dv\bigg|\\
 & \le|\lambda|M_{3}\int_{t}^{T}|\lambda|M_{3}\max_{t\le s\le T}|k_{2}(s)-l_{2}(s)|(T-v)dv\\
 & =\frac{1}{2}|\lambda|^{2}M_{3}^{2}(T-t)^{2}\max_{t\le s\le T}|k_{2}(s)-l_{2}(s)|.
\end{align*}
By induction, we have
\begin{align*}
|F^{n}k_{2}(t)-F^{n}l_{2}(t)| & \le\frac{1}{n!}|\lambda|^{n}M_{3}^{n}(T-t)^{n}\max_{t\le s\le T}|k_{2}(s)-l_{2}(s)|\\
 & \le\frac{1}{n!}|\lambda|^{n}M_{3}^{n}T^{n}\max_{t\le s\le T}|k_{2}(s)-l_{2}(s)|.
\end{align*}
Since
\[
\lim_{n\to\infty}\frac{\{|\lambda|M_{3}T\}^{n}}{n!}=0,
\]
therefore, for the given fixed values of $\lambda$, $M_{3}$ and
$T$, there exists an integer $N$ such that
\[
0\le\frac{\{|\lambda|M_{3}T\}^{N}}{N!}<1.
\]
So the mapping $F^{N}$ is a contraction and thus equation (\ref{eq:Volterra_equation})
has one and only one solution.

\section{}\label{sec: Appendix_C}

\setcounter{equation}{0} \renewcommand{\theequation}{C.\arabic{equation}}

\textbf{Proof of Theorem \ref{thm:Convergence_Speed_k1_k2}. }The
proof of the convergence speed of the iterative scheme \eqref{eq:Iteration_Scheme_k10}--\eqref{eq:Iteration_Scheme_k1n}
is similar to Theorem 4.9 of Björk et al. (2014). We thus omit it
here.

Using the notations in Appendix \ref{sec: Appendix_B}, the iteration
scheme \eqref{eq:Iteration_Scheme_k20}--\eqref{eq:Iteration_Scheme_k2n}
can be written as
\[
k_{2}^{(n)}(t)=\phi(t)+\lambda\int_{t}^{T}L(t,v)k_{2}^{(n-1)}(v)dv.
\]

Denote $\bar{k}_{2}^{(n)}=k_{2}^{(n)}-k_{2}^{(n-1)},$ for $\forall t\in[0,T]$,
we have
\begin{align}
\big|\bar{k}_{2}^{(n)}(t)\big| & =|\lambda|\big|\int_{t}^{T}L(t,v)\big(k_{2}^{(n-1)}(v)-k_{2}^{(n-2)}(v)\big)dv\big|\nonumber \\
 & \le|\lambda|M_{3}\int_{t}^{T}\big|\bar{k}_{2}^{(n-1)}(v)\big|dv.\label{eq:(4.16)_in_Zhou}
\end{align}
Let $\omega_{n}(t)=\int_{t}^{T}\big|\bar{k}_{2}^{(n)}(s)\big|ds$.
From Equation \eqref{eq:(4.16)_in_Zhou}, we have
\[
\frac{d}{dt}\omega_{n}(t)+|\lambda|M_{3}\omega_{n-1}(t)\ge0,
\]
and thus
\begin{align*}
\omega_{n}(t) & \le|\lambda|M_{3}\int_{t}^{T}\omega_{n-1}(v)dv\\
 & \le(|\lambda|M_{3})^{2}\int_{t}^{T}\int_{v}^{T}\omega_{n-1}(s)dsdv\\
 & =(|\lambda|M_{3})^{2}\int_{t}^{T}\int_{t}^{s}\omega_{n-1}(s)dvds\\
 & =(|\lambda|M_{3})^{2}\int_{t}^{T}(s-t)\omega_{n-1}(s)dv\\
 & \le(|\lambda|M_{3})^{3}\int_{t}^{T}(s-t)\int_{s}^{T}\omega_{n-2}(v)dvds\\
 & \le(|\lambda|M_{3})^{3}\int_{t}^{T}\int_{t}^{v}(s-t)\omega_{n-2}(v)dsdv\\
 & \le(|\lambda|M_{3})^{3}\int_{t}^{T}\frac{1}{2!}(v-t)^{2}\omega_{n-2}(v)dv\\
 & \le\cdots\\
 & \le(|\lambda|M_{3})^{n}\int_{t}^{T}\frac{1}{(n-1)!}(v-t)^{n-1}\omega_{1}(v)dv\\
 & \le\frac{1}{n!}(|\lambda|M_{3})^{n}(T-t)^{n}\omega_{1}(0).
\end{align*}
Therefore,
\begin{align*}
\big|k_{2}^{(n)}(t)-k_{2}(t)\big| & =\big|-\sum_{j=n}^{\infty}\bar{k}_{2}^{(j+1)}\big|\\
 & \le\sum_{j=n}^{\infty}\big|\bar{k}_{2}^{(j+1)}\big|\\
 & \le\sum_{j=n}^{\infty}\frac{1}{j!}\omega_{1}(0)(|\lambda|M_{3})^{j+1}(T-t)^{j}\\
 & \le\sum_{j=n}^{\infty}\frac{1}{j!}K{}^{j+1}(T-t)^{j}
\end{align*}
for $n=1,\ldots.$ Here, $K$ can be selected as any positive constant
larger than $\max(|\lambda|M_{3},\omega_{1}(0)|\lambda|M_{3})$.

\end{appendices}

\section*{References}
\begin{description}
\item [{\textmd{Bajeux-Besnainou,}}] I., and Portait, R. (1998). Dynamic
asset allocation in a mean--variance framework. \textit{Management
Science}, 44(11), 79--95.
\item [{\textmd{Basak,}}] S. and Chabakauri, G. (2010). Dynamic mean--variance
asset allocation. \textit{Review of Financial Studies}, 23, 2970--3016.
\item [{\textmd{Bielecki,}}] T. R., Jin, H., Pliska, S. R., and Zhou, X.
Y. (2005). Continuous time mean variance portfolio selection with
bankruptcy prohibition. \textit{Mathematical Finance}, 15(2), 213--244.
\item [{\textmd{Björk,}}] T., and Murgoci, A. (2010). A general theory
of Markovian time inconsistent stochastic control problems. Working
Paper, Stockholm School of Economics. Available at SSRN 1694759.
\item [{\textmd{Björk,}}] T., and Murgoci, A. (2014). A theory of Markovian
time-inconsistent stochastic control in discrete time. \textit{Finance
and Stochastics}, 18(3), 545--592.
\item [{\textmd{Björk,}}] T., Murgoci, A., and Zhou, X. Y. (2014). Mean--variance
portfolio optimization with state-dependent risk aversion. \textit{Mathematical
Finance}, 24(1), 1--24.
\item [{\textmd{Bensoussan,}}] A., Wong, K. C., Yam, S. C. P., and Yung,
S. P. (2014). Time-consistent portfolio selection under short-selling
prohibition: from discrete to continuous setting. \textit{SIAM Journal
on Financial Mathematics}, 5(1), 153--190.
\item [{\textmd{Chen,}}] P., and Yang, H. (2011). Markowitz's mean--variance
asset-liability management with regime switching: a multi-period model.
\textit{Applied Mathematical Finance}, 18(1), 29--50.
\item [{\textmd{Chen,}}] P., Yang, H., and Yin, G. (2008). Markowitz\textquoteright s
mean--variance asset-liability management with regime switching: a
continuous-time model. \textit{Insurance: Mathematics and Economics},
43(3), 456--465.
\item [{\textmd{Dai,}}] M., Xu, Z. Q., and Zhou, X. Y. (2010). Continuous-time
Markowitz's model with transaction costs. \textit{SIAM Journal on
Financial Mathematics}, 1(1), 96--125.
\item [{\textmd{Ekeland,}}] I., and Lazrak, A. (2006). Being serious about
non-commitment: subgame perfect equilibrium in continuous time. Preprint.
University of British Columbia.
\item [{\textmd{Ekeland,}}] I., and Pirvu, T. A. (2008). Investment and
consumption without commitment. \textit{Mathematics and Financial
Economics}, 2(1), 57--86.
\item [{\textmd{Konno,}}] H., and Yamazaki, H. (1991). Mean-absolute deviation
portfolio optimization model and its applications to Tokyo stock market.
\textit{Management Science}, 37(5), 519--531.
\item [{\textmd{Li,}}] Y., and Li, Z. (2013). Optimal time-consistent investment
and reinsurance strategies for mean--variance insurers with state
dependent risk aversion. \textit{Insurance: Mathematics and Economics},
53(1), 86--97.
\item [{\textmd{Li,}}] D., and Ng, W. L. (2000). Optimal dynamic portfolio
selection: multi-period mean--variance formulation. \textit{Mathematical
Finance}, 10(3), 387--406.
\item [{\textmd{Li,}}] D., Rong, X., and Zhao, H. (2015). Time-consistent
reinsurance-investment strategy for a mean\textendash variance insurer
under stochastic interest rate model and inflation risk. \textit{Insurance:
Mathematics and Economics}, 64, 28--44.
\item [{\textmd{Lim,}}] A. E. (2004). Quadratic hedging and mean--variance
portfolio selection with random parameters in an incomplete market.
\textit{Mathematics of Operations Research}, 29(1), 132--161.
\item [{\textmd{Lim,}}] A. E., and Zhou, X. Y. (2002). Mean-variance portfolio
selection with random parameters in a complete market. \textit{Mathematics
of Operations Research}, 27(1), 101--120.
\item [{\textmd{Lin,}}] X., and Qian, Y. (2015). Time-consistent mean--variance
reinsurance-investment strategy for insurers under CEV model. \textit{Scandinavian
Actuarial Journal}, 1--26.
\item [{\textmd{Markowitz,}}] H. (1952). Portfolio selection. \textit{The
Journal of Finance}, 7(1), 77--91.
\item [{\textmd{Markowitz,}}] H. (1956). The optimization of a quadratic
function subject to linear constraints. \textit{Naval Research Logistics
Quarterly}, 3, 111--133.
\item [{\textmd{Markowitz,}}] H. (1959). \textit{Portfolio Selection: Efficient
Diversification of Investments}. New York: John Wiley and Sons.
\item [{\textmd{Markowitz,}}] H., Todd, P., Xu, G., and Yamane, Y. (1993).
Computation of mean--semivariance efficient sets by the critical line
algorithm. \textit{Annals of Operations Research}, 45(1), 307--317.
\item [{\textmd{Merton,}}] R. C. (1972). An analytical derivation of the
efficient portfolio frontier. \textit{Journal of Financial and Quantitative
Analysis}, 7, 1851--1872.
\item [{\textmd{Ogryczak,}}] W., and Ruszczy\'{n}ski, A. (1999). From stochastic
dominance to mean--risk models: semideviations as risk measures. \textit{European
Journal of Operational Research}, 116(1), 33--50.
\item [{\textmd{Peleg,}}] B., and Yaari, M. E. (1973). On the existence
of a consistent course of action when tastes are changing. \textit{The
Review of Economic Studies}, 391--401.
\item [{\textmd{Perold,}}] A. F. (1984). Large-scale portfolio optimization.
\textit{Management Science}, 30(10), 1143--1160.
\item [{\textmd{Pollak,}}] R. A. (1968). Consistent planning. \textit{The
Review of Economic Studies}, 201--208.
\item [{\textmd{Richardson,}}] H. R. (1989). A minimum variance result
in continuous trading portfolio optimization. \textit{Management Science},
35(9), 1045--1055.
\item [{\textmd{Strotz,}}] R. H. (1955). Myopia and inconsistency in dynamic
utility maximization. \textit{The Review of Economic Studies}, 23,
165--180.
\item [{\textmd{Wang,}}] J., and Forsyth, P. A. (2010). Numerical solution
of the Hamilton--Jacobi--Bellman formulation for continuous time mean
variance asset allocation. \textit{Journal of Economic Dynamics and
Control}, 34(2), 207--230.
\item [{\textmd{Wang,}}] J., and Forsyth, P. A. (2011). Continuous time
mean variance asset allocation: a time-consistent strategy. \textit{European
Journal of Operational Research}, 209(2), 184--201.
\item [{\textmd{Wang,}}] J., and Forsyth, P. A. (2012). Comparison of mean
variance like strategies for optimal asset allocation problems. \textit{International
Journal of Theoretical and Applied Finance}, 15(02), 1250014.
\item [{\textmd{Wei,}}] J., Wong, K. C., Yam, S. C. P., and Yung, S. P.
(2013). Markowitz\textquoteright s mean--variance asset-liability
management with regime switching: a time-consistent approach. \textit{Insurance:
Mathematics and Economics}, 53(1), 281--291.
\item [{\textmd{Xia,}}] J. (2005). Mean--variance portfolio choice: quadratic
partial hedging. \textit{Mathematical Finance}, 15(3), 533--538.
\item [{\textmd{Zeng,}}] Y., Li, Z., and Lai, Y. (2013). Time-consistent
investment and reinsurance strategies for mean--variance insurers
with jumps. \textit{Insurance: Mathematics and Economics}, 52(3),
498--507.
\item [{\textmd{Zenios,}}] S. A., and Kang, P. (1993). Mean--absolute deviation
portfolio optimization for mortgage-backed securities. \textit{Annals
of Operations Research}, 45(1), 433--450.
\item [{\textmd{Zhou,}}] X. Y., and Li, D. (2000). Continuous-time mean--variance
portfolio selection: a stochastic LQ framework. \textit{Applied Mathematics
and Optimization}, 42(1), 19--33.
\item [{\textmd{Zhou,}}] X. Y., and Yin, G. (2003). Markowitz's mean--variance
portfolio selection with regime switching: a continuous-time model.
\textit{SIAM Journal on Control and Optimization}, 42(4), 1466--1482.\end{description}

\end{document}